\documentclass{article}[numberwithinsect]
\usepackage{fullpage}
\usepackage{scontents}
\newenvsc{append}[store-env=append,print-env=true]
\usepackage[linesnumbered,ruled,vlined]{algorithm2e}
\usepackage{mathtools,amssymb}
\usepackage{cite}
\usepackage{placeins}
\usepackage{eqparbox}
\usepackage{amsthm}
\usepackage{amsmath}
\usepackage{comment}
\usepackage{amssymb}
\usepackage{hyperref}
\usepackage{amsmath}
\usepackage{nicefrac}
\usepackage{xcolor}
\usepackage[colorinlistoftodos]{todonotes}
\usepackage[linesnumbered,ruled,vlined]{algorithm2e}
\usepackage{float}
\usepackage{xspace}
\usepackage{cleveref}
\usepackage{MnSymbol,wasysym}
\usepackage{subfiles}
\usepackage{thm-restate}
\usepackage{subfiles}

\newtheorem{theorem}{Theorem}[section]
\newtheorem{lemma}[theorem]{Lemma}
\newtheorem{definition}[theorem]{Definition}
\newtheorem{claim}[theorem]{Claim}
\newtheorem{corollary}[theorem]{Corollary}
\newtheorem{observation}[theorem]{Observation}
\newtheorem{remark}[theorem]{Remark}

\newcommand{\set}[1]{\left\{ #1 \right\}}
\newcommand{\mc}{Max-Cut\xspace}
\newcommand{\mdc}{Max-DiCut\xspace}
\newcommand{\mSAT}{Max-SAT\xspace}
\newcommand{\mtwoSAT}{Max-$2$SAT\xspace}
\newcommand{\mbisec}{Max-Bisection\xspace}
\newcommand{\cn}{Cut Norm\xspace}
\newcommand{\ma}{Max-Agreement\xspace}
\newcommand{\vc}{Vertex Cover\xspace}

\newcommand{\ftc}{$AFTcut$\xspace}
\newcommand{\aftv}{FT value\xspace}
\newcommand{\akftv}{$k$-FT value\xspace}

\newcommand{\oftv}{$k$-FT value\xspace}

\newcommand{\la}{\leftarrow}

\newcommand{\sm}{-}
\newcommand{\asmc}{\alpha_{SMC}}
\newcommand{\agw}{\alpha_{GW}}
\newcommand{\kftc}{$k$-$AFTcut$\xspace}
\newcommand{\smc}{Simultaneous Max-Cut\xspace}
\newcommand{\xftc}[1]{$#1$-$AFTcut$\xspace}

\newcommand{\oblftc}{$OFTcut$\xspace}
\newcommand{\koblftc}{$k$-$OFTcut$\xspace}
\newcommand{\mexp}[1]{\mathop{\mathbb{E}}\left[#1\right]}
\newcommand{\mexpl}[2]{\underset{#1}{\mathop{\mathbb{E}}}\left[#2\right]}
\newcommand{\mpr}[1]{\mathop{\mathbb{P}}\left[#1\right]}

\newcommand{\sv}{S\oplus v}
\newcommand{\md}{\ell} 
\newcommand{\sol}{\widetilde{S}}

\newcommand{\pstar}{{C_{S^*_{smc}-k\times H}}}
\newcommand{\ptilde}{{C_{\sol-k\times H}}}

\newcommand \X {\hat H}

\DeclareMathOperator*{\argmax}{arg\,max}

\title{Fault Tolerant Max-Cut\thanks{This project has received funding from the European Research Council (ERC), under the European Unions Horizon 2020 research and innovation programme under grant agreement No 755839, and from Israel Science Foundation (ISF), under grant No 1336/16.}}
\author{Keren Censor-Hillel\\
Technion\\
\texttt{ckeren@cs.technion.ac.il}
\and
Noa Marelly\\
Technion\\
\texttt{noa.marelly@gmail.com}
\and
Roy Schwartz\\
Technion\\
\texttt{schwartz@cs.technion.ac.il}
\and
Tigran Tonoyan\\
Technion\\
\texttt{ttonoyan@gmail.com}
}
\date{}

\begin{document}
\maketitle

\begin{abstract}

In this work, we initiate the study of fault tolerant \mc, where given an edge-weighted undirected graph $G=(V,E)$, the goal is to find a cut $S\subseteq V$ that maximizes the total weight of edges that cross $S$ even after an adversary removes $k$ vertices from $G$. 
We consider two types of adversaries: an adaptive adversary that sees the outcome of the random coin tosses used by the algorithm, and an oblivious adversary that does not.
For any constant number of failures $k$ we present an approximation of $(0.878-\epsilon)$ against an adaptive adversary and of $\agw\approx 0.8786$ against an oblivious adversary (here $\agw$ is the approximation achieved by the random hyperplane algorithm of [Goemans-Williamson J. ACM `95]).
Additionally, we present a hardness of approximation of $ \agw$ against both types of adversaries, rendering our results (virtually) tight.

The non-linear nature of the fault tolerant objective makes the design and analysis of algorithms harder when compared to the classic \mc.
Hence, we employ approaches ranging from multi-objective optimization to LP duality and the ellipsoid algorithm to obtain our results.

\end{abstract}

\section{Introduction}
In this work, we initiate the study of \emph{fault tolerant \mc}.
In the classic \mc problem, we are given an undirected graph $G=(V,E)$ equipped with non-negative edge weights $w:E\rightarrow \mathbb{R}_+$.
The goal is to find a cut $ S\subseteq V$ that maximizes the total weight of edges that cross $S$.
\mc is one of Karp’s 21 NP-complete problems \cite{Karp_np_complete} and has been for close to three decades a case study for the introduction of new approaches both in the theory of algorithms and the complexity theory.
Perhaps the two most prominent examples of the above are: $(1)$ the random hyperplane rounding method of Goemans and Williamson for semi-definite programs \cite{DBLP:journals/jacm/GoemansW95}, which yields an approximation of $\agw\approx 0.8786$ for \mc; and $(2)$ the Unique Games Conjecture of Khot \cite{DBLP:conf/stoc/Khot02a}.
The former has opened an entirely new area in the field of approximation algorithms with applications to a wide range of problems, {\em e.g.}, \mdc \cite{feige1995approximating,DBLP:conf/random/MatuuraM01, DBLP:conf/ipco/LewinLZ02}, \mbisec \cite{AustrinBG16,RaghavendraT12}, \ma \cite{CharikarGW05,Swamy04}, \mtwoSAT \cite{feige1995approximating,DBLP:conf/ipco/LewinLZ02}, \mSAT \cite{asano2002improved,DBLP:conf/waoa/AvidorBZ05}, and \cn \cite{AlonNaor}, to name a few.
The latter has been a dominant method for proving hardness of approximation results in the last two decades, {\em e.g.}, the celebrated tight hardness for \mc \cite{DBLP:journals/siamcomp/KhotKMO07,MOO}, and \vc \cite{DBLP:journals/jcss/KhotR08}.

Motivated by large scale real life systems, fault tolerant algorithms seek to find a solution to a given optimization problem that is resilient to failures of some parts of the input.
The above can be intuitively formulated as a two step process: $(1)$ the algorithm finds a solution to the problem at hand; and $(2)$ an adversary removes parts of the input.
The goal of the algorithm is that no matter which part of the input the adversary removes, the remaining solution after removal still retains some desired properties despite the removal.
Typically, the focus of fault tolerance has been network design problems, {\em e.g.}, 
BFS \cite{DBLP:journals/corr/abs-1302-5401, DBLP:conf/soda/ParterP14, DBLP:conf/podc/Parter15, DBLP:conf/spaa/ParterP15, DBLP:conf/icalp/GuptaK17} and spanners \cite{DBLP:journals/algorithmica/LevcopoulosNS02, DBLP:conf/podc/DinitzK11, DBLP:conf/stoc/Solomon14, DBLP:conf/wdag/Parter14,DBLP:conf/soda/BodwinDPW18, DBLP:conf/podc/BodwinP19, DBLP:conf/icalp/BodwinGPW17}.
Additional related algorithmic problems for which fault tolerant algorithms were studied include, {\em e.g.}, single source reachability \cite{DBLP:conf/wdag/BaswanaCR15, DBLP:conf/stoc/BaswanaCR16}, connected dominating set \cite{DBLP:journals/informs/BuchananSBP15, DBLP:journals/informs/ZhouZTHD18}, and facility location \cite{DBLP:journals/algorithmica/JainV03, DBLP:journals/jal/GuhaMM03, DBLP:journals/talg/SwamyS08, DBLP:journals/tcs/ChechikP14}.

In this work, we initiate the study of fault tolerant \mc, where the adversary can remove vertices from the graph (all edges touching the removed vertices are also deleted).
Intuitively, fault tolerant \mc can be seen as a two players game, in which one player (the algorithm) chooses a cut and the other player (the adversary) removes up to a prespecified number $k$ of vertices.
The algorithm desires to  maximize the total weight of edges crossing the cut, while the adversary aims to minimize the total weight of edges crossing the cut.

We study two types of adversaries.
The first is an \emph{adaptive adversary} that chooses which $k$ vertices to fail \emph{after} seeing the cut the algorithm produces.
Specifically, the adaptive adversary knows the input, how the algorithm operates, and if the algorithm is randomized, the adaptive adversary also knows the outcome of all random coin tosses used by the algorithm. The second type of adversary is an \emph{oblivious adversary}.
Similarly to the adaptive adversary, the oblivious adversary knows the input and how the algorithm operates.
However, in contrast to the adaptive adversary, the oblivious adversary does not know the outcome of the random coin tosses used by the algorithm, in case the latter is randomized (equivalently, the oblivious adversary only knows the distribution over cuts the algorithm produces).
Thus, the oblivious adversary is required to choose which $k$ vertices to fail without the knowledge of which cut was sampled.
To the best of our knowledge only adaptive adversaries were studied in the fault tolerance literature.

~\\\textbf{The Challenges.} 
The fault tolerant \mc problem differs considerably from classic \mc for several reasons.
First, the structure of the solutions may be different.
Specifically, there are instances for which an optimal solution to fault tolerant \mc is not an optimal solution to classic \mc, and vice versa.
Furthermore, it might be the case that the ratio between the values of the optimal solutions is large or even unbounded. 

Second, the application of known techniques (which can be successfully applied to  \mc) to fault tolerant \mc imposes some obstacles that arise from the non-linear nature of the fault tolerant objective.
For example, the random hyperplane rounding method of Goemans and Williamson cannot be analyzed in a straightforward manner as one is required to lower bound the expectation of the minimum value (over all possible actions of the adversary) of the cut the random hyperplane defines, as opposed to just the expected value of the cut the random hyperplane defines.
Moreover, even analyzing the simplest known algorithm for \mc, {\em i.e.}, choosing a uniform random cut, requires great care (refer to Section \ref{sec:randomcut} for further details).
Hence, the design and analysis of algorithms for fault tolerant \mc requires some new insights into the problem.

\subsection{Our Contributions}

\subparagraph*{Adaptive Adversary.}

When focusing on an adaptive adversary, our main result is an (almost) tight approximation of $0.878-\epsilon$, for any constant number $k$ of failures and unweighted graphs.
This is summarized in the following theorem (it is important to note that the constant in the  theorem is slightly smaller than the Goemans-Williamson approximation factor $\agw$).
\begin{restatable}{theorem}{thmkftc}\label{thm:kftc} For every  constant $k > 0$ and $\epsilon > 0$, there is a polynomial time $(0.878-\epsilon)$-approximation algorithm for fault tolerant \mc  on unweighted graphs against an adaptive adversary and $k$ faults.
\end{restatable}
Our algorithm is based on viewing fault tolerant \mc against an adaptive adversary as a multi-objective optimization problem, where for every possible subset of $k$ vertices the adversary can fail, one can define a different objective.
The goal is to maximize the worst, {\em i.e.}, minimum, objective.
This approach does not suffice, since all known results for the multi-objective variant of \mc (formally known as \smc \cite{DBLP:conf/icalp/BhangaleKS15, DBLP:conf/soda/BhangaleKKST18}) can handle only a constant number of objectives.
In our case, even when a single failure is allowed, the number of objectives equals $n$.
Hence, to overcome this difficulty, we incorporate local search into the above multi-objective approach to obtain the claimed result in Theorem~\ref{thm:kftc}.

\subparagraph*{Oblivious Adversary.}
When focusing on an oblivious adversary, our main result is a tight approximation of $\agw$ for any constant number $k$ of failures.
However, in contrast to the adaptive adversary setting, this result holds for general weighted graphs and achieves the $\agw$-approximation guarantee exactly.
This is summarized in the following theorem.
\begin{restatable}{theorem}{thmoblagw}\label{thm:obl_agw}
For every constant $k>0$, there is a polynomial time $\agw$-approximation algorithm for fault tolerant \mc on general weighted graphs against an oblivious adversary and  $k$ faults.
\end{restatable}
The approach we adopt for approximating fault tolerant \mc against an oblivious adversary significantly differs from the approach taken against an adaptive adversary.
Surprisingly, our algorithm is based on an approximation-preserving reduction from fault tolerant \mc to the classic \mc problem.
This reduction uses LP duality alongside the ellipsoid algorithm and is achieved by presenting a suitable approximate dual separation oracle for a configuration LP that encodes the distribution over cuts that the algorithm produces.

\subparagraph*{Hardness of Approximation.}
We prove that fault tolerant \mc in unweighted graphs, against both adaptive and oblivious adversaries, cannot be approximated better than $\agw$ without breaking well-known hardness assumptions.
It is important to note that this settles the approximability of the oblivious adversary setting  (see Theorem \ref{thm:obl_agw} above), and almost settles the approximability of the adaptive adversary setting (see Theorem \ref{thm:kftc} above) as the constant in Theorem \ref{thm:kftc} is slightly smaller than $\agw$.

\begin{restatable}{theorem}{hardnessobl}\label{thm:hardness_obv}
Assuming the Unique Games Conjecture and $NP\nsubseteq BPP$, there is no polynomial time  
$(\agw+\epsilon)$-approximation algorithm for fault tolerant \mc in unweighted graphs, for any constant $\epsilon>0$. This holds for both adaptive and oblivious adversaries.
\end{restatable}

\subparagraph*{Simple Purely Combinatorial Algorithms.}
While Theorem \ref{thm:kftc} provides an (almost) tight result against an adaptive adversary, and Theorem \ref{thm:obl_agw} provides a tight result against an oblivious adversary, the techniques we employ yield algorithms which are polynomial but not simple.
For example, the work of \cite{DBLP:conf/soda/BhangaleKKST18} for approximating \smc, an important ingredient in the design of our algorithm against an adaptive adversary, is based on SDP hierarchies and the running time is exponential in the number of objectives.
In contrast, the classic \mc problem admits some very simple and fast heuristics, {\em e.g.}, choosing a random uniform cut.
Thus, we also aim to study simple and purely combinatorial algorithms for fault tolerant \mc.

We prove that fault tolerant \mc does yield a simple purely combinatorial local search algorithm with a provable approximation guarantee against an adaptive adversary.
Unfortunately, the classic local search for \mc, that in each step moves a single vertex from one side of the cut to the other side, fails in the fault tolerant setting.
Nonetheless, we prove that a local search that allows for a slightly  richer family of local improvement steps suffices.
This is summarized in the following theorem (refer to Section \ref{sec:CombLS} for additional details).
\begin{restatable}{theorem}{greedyftc}\label{thm:Greedy_ftc}
 There 
 is a purely  combinatorial polynomial time $1/2$-approximation algorithm 
 for fault tolerant \mc 
 on unweighted input graphs against an adaptive adversary and a single fault.
\end{restatable}
We further study how a \emph{uniform random cut} performs against both types of adversaries, and prove that this performance depends on the type of the adversary.
Specifically, for an oblivious adversary an approximation of $1/2$ is achieved, by a uniform random cut.
However, this is not the case when considering an adaptive adversary, since we prove that a uniform random cut cannot achieve an approximation better than $1/4$.

\subsection{Related Work}

The weighted version of \mc is one of Karp's NP-complete problems \cite{Karp_np_complete}, and the unweighted version is also known to be NP-complete \cite{DBLP:journals/tcs/GareyJS76}. In general graphs, one cannot obtain an approximation factor better than $16/17$ for the undirected version, or better than $12/13$ for the directed version, unless $P=NP$ \cite{DBLP:journals/siamcomp/TrevisanSSW00, DBLP:journals/jacm/Hastad01}. The best known approximation for \mc is the celebrated  random hyperplane algorithm of Goemans and Williamson that obtains an approximation factor of roughly  $0.8786$  by rounding the natural semi-definite programming relaxation \cite{DBLP:journals/jacm/GoemansW95}. This is the best approximation that one can achieve, assuming the Unique Games Conjecture of Khot \cite{DBLP:journals/siamcomp/KhotKMO07} and $P\neq NP$.

The problem of fault tolerant \mc against an adaptive adversary that we introduce in this paper can be viewed as a special case of \smc, in which the input is a collection of $\tau$ weighted graphs on the same vertex set and the goal is to partition the vertices into two parts, such that the size of the cut is large in every given graph. In a straightforward manner, our problem would imply $\tau=\binom{n}{k}$, which is unacceptable since the known approximations for \smc are for a constant number of instances only \cite{DBLP:journals/dam/AngelBG06, DBLP:conf/icalp/BhangaleKS15, DBLP:conf/soda/BhangaleKKST18}. Nonetheless, we do use the algorithm from \cite{DBLP:conf/soda/BhangaleKKST18} to obtain an algorithm that achieves an approximation of $0.878$ for fault tolerant \mc against an adaptive adversary. The state-of-the-art for \smc is a polynomial $0.878$-approximation for any  constant number of input graphs \cite{DBLP:conf/soda/BhangaleKKST18},  which is nearly optimal since assuming the Unique Games conjecture, \smc cannot be approximated better than $ (\agw -\delta)$ (where $\delta \geq 10^{-5}$) \cite{DBLP:conf/coco/BhangaleK20}.

One more notion of resilience is that of robust submodular maximization, see, {\em e.g.}, \cite{DBLP:conf/ipco/OrlinSU16, DBLP:conf/kdd/AvdiukhinMYZ19}. Given a submodular function $f$ and, {\em e.g.}, a cardinality constraint $k$, a set $A$ is robust against $\tau$ failures if $A = \argmax_{A\subseteq V, |A|\leq k}{\min_{Z\subseteq A, |Z|\leq \tau}{f(A\sm Z)}}$, i.e., a subset of size at most $k$ that achieves the maximal value after at most $\tau$ elements are removed from the solution. Note that this notion of robustness differs from fault tolerance.
The reason is that the failed elements are removed from the solution, as opposed to removed from the instance.
Specifically, when considering the cut function of an undirected graph (which is submodular) the removal of a vertex from $S$ (as in robust) differs from removing the same vertex from the graph (as in fault tolerant).

Due to the importance of coping with failures, the fault tolerance of many additional fundamental problems has been extensively studied. Prime examples are replacement paths \cite{DBLP:conf/icalp/AlonCC19, DBLP:conf/soda/ChechikC19, DBLP:conf/icalp/ChechikM20, DBLP:conf/focs/GrandoniW12, DBLP:journals/talg/RodittyZ12}, BFS trees \cite{DBLP:journals/corr/abs-1302-5401, DBLP:conf/soda/ParterP14, DBLP:conf/podc/Parter15, DBLP:conf/spaa/ParterP15, DBLP:conf/icalp/GuptaK17}, spanners \cite{DBLP:journals/algorithmica/LevcopoulosNS02, DBLP:conf/podc/DinitzK11, DBLP:conf/stoc/Solomon14, DBLP:conf/wdag/Parter14,DBLP:conf/soda/BodwinDPW18, DBLP:conf/podc/BodwinP19, DBLP:conf/icalp/BodwinGPW17}, connected dominating sets \cite{DBLP:journals/informs/BuchananSBP15, DBLP:journals/informs/ZhouZTHD18},  and more \cite{DBLP:journals/algorithmica/JainV03, DBLP:journals/jal/GuhaMM03, DBLP:journals/talg/SwamyS08, DBLP:journals/tcs/ChechikP14, DBLP:conf/stoc/BaswanaCR16, DBLP:journals/algorithmica/BiloGLP18, DBLP:conf/wdag/BaswanaCR15, DBLP:conf/soda/BaswanaCHR18}

Fault tolerance was also studied in the distributed setting, such as for BFS trees \cite{DBLP:conf/spaa/GhaffariP16}, MST \cite{DBLP:conf/spaa/GhaffariP16}, and spanners \cite{DBLP:conf/podc/DinitzK11, DBLP:conf/wdag/Parter20}.

\subparagraph*{Paper Organization.}
Section \ref{sec:prelim} contains all required formal definitions and preliminary lemmas used throughout the paper.
Section \ref{sec:Adaptive} deals with the adaptive adversary, whereas Section \ref{sec:Oblivious} deals with the oblivious adversary. Section~\ref{sec:hardness} consists of the hardness results, and Section~\ref{sec:randomcut} deals with a uniform random cut.

\section{Preliminaries}\label{sec:prelim}

\subparagraph*{Graph Notations.} We  consider only edge-weighted graphs $G=(V,E,w)$ with positive integer weights $w_e$ assigned to the edges $e\in E$. By \emph{unweighted} graphs we mean graphs with $w_e=1$, for all $e\in E$. A \emph{cut} $S$ in a graph $G=\left(V,E, w\right)$ is a subset of vertices $S\subseteq V$. 
We let $\delta(S, G)=\{e\in E : |e\cap S|=1\}$ denote the set of all \emph{crossing edges} of $S$ in the graph $G$.
The \emph{size} or \emph{weight} of a cut $S$, denoted by $C_{S,G}$, is the total weight of the crossing edges: $C_{S,G} = \sum_{e \in \delta(S,G)}{w_{e}}$. When $G$ is clear from the context, we use $C_S$ and $\delta(S)$.

For a set $F\subseteq V$ of vertices, the \emph{degree} $d(F)$  of $F$  is the total weight of edges adjacent to $F$: $d(F)=\sum_{e\in E: e\cap F\neq \emptyset}w_e$. For a subset $F\subseteq V$ and cut $S\subseteq V$, the \emph{crossing degree} $d_S(F)$ of $F$  is the total weight of edges adjacent to $F$ that cross $S$: $d_S(F)=\sum_{e\in \delta(S): e\cap F\neq \emptyset}w_e$. We use $d(v)$ and $d_S(v)$, if $F=\{v\}$. We also let $n = |V|$, $m = |E|$, and $\Delta = \max_{v\in V}{d(v)}$. Finally, we let $2^V$ and ${\binom{V}{k}}$ denote the collection of all and all size-$k$ subsets of $V$, respectively.

\subparagraph*{The Adaptive Adversary.}
We define the \emph{\akftv of a cut against an adaptive adversary} to be the minimal size of the cut, subsequent to a failure of any $k$ vertices. Formally, for a cut $S$ in a graph $G = (V,E, w)$ and a constant $k> 0$, the \akftv of $S$ is defined as $\varphi(S, k, G) = \min_{F\in \binom{V}{k}}{C_{S-F, G-F}}$.

\begin{definition}[\kftc]
Given an edge-weighted graph $G=(V,E,w)$ and a number $k\in \mathbb{N}$, a cut $S$ is a \emph{$k$-adaptive fault tolerant cut}, or \kftc for short, if $\varphi\left(S, k, G \right)=\underset{S'\subseteq V}{\max}\set{\varphi\left( S', k,G\right)}$. 
\end{definition}

We usually omit $G$ and/or $k$ from $\varphi(S,k,G)$ when $G$ is clear from the context and $k=1$. 
The \emph{\mc} problem, i.e., that of finding a cut with the largest size, corresponds to the special case $k=0$.

\subparagraph*{The Oblivious Adversary.}
We represent a randomized algorithm that finds a cut in a graph $G=(V,E,w)$ by a probability distribution $\mathcal{D}$ over all possible cuts $2^V$. 
For a distribution $\mathcal{D}$ over cuts, we define the \emph{\oftv of $\mathcal{D}$} to be the minimal expected size of the cut, subsequent to the failure of any $k$ vertices. Formally, for a graph $G = (V,E, w)$, a distribution $\mathcal{D}$ over cuts  and a constant $k> 0$, we define the \oftv of $\mathcal{D}$, denoted by $\mu(\mathcal{D}, k,G)$, as $\mu(\mathcal{D}, k, G) = \min_{F\in \binom{V}{k}}{\mexpl{S\sim \mathcal{D}}{C_{S-F, G-F}}}$.

\begin{definition}[\koblftc]
Given an edge-weighted graph $G=(V,E,w)$ and a number $k\in \mathbb{N}$, a distribution $\mathcal{D}$ over all cuts $2^V$ is a $k$-oblivious fault tolerant cut, or \koblftc for short, if $\mu(\mathcal{D}, k, G) = \underset{\mathcal{D'}}{\max}\set{\mu(\mathcal{D'}, k, G)}$. 
\end{definition}

Note that here we assume the adversary chooses the set $F$ of faults deterministically; it easily follows from the linearity of expectation that the adversary always has a deterministic best choice -- a subset that has the largest expected crossing degree.

\subparagraph{Dissimilarity of \ftc and \mc.} In the two following observations we show that the problem of finding an \ftc differs from finding a \mc, that is, there exists a solution for \mc which is not a solution for \ftc, and vice versa.

\begin{figure}[h]
\centering
\includegraphics[width=8.5cm, scale=0.25]{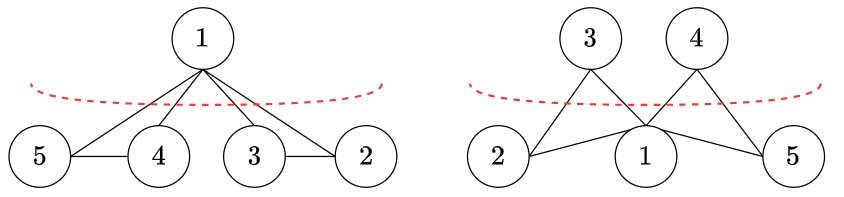}
\caption{A \mc that is not an \ftc}
\label{figure:mc_no_ftc}
\end{figure}

\begin{observation}{There exists a solution for \mc in $G$, which is not a solution for \ftc.}
\end{observation}
\begin{proof}
Consider a graph that consists two triangles that share a vertex, and label the shared vertex by $1$ (see Figure~\ref{figure:mc_no_ftc}).
It holds that $\set{1}$ is a \mc. In addition, $\varphi(\set{1}) = 0$ because when vertex $1$ fails there are no crossing edges. There is a better solution for \ftc, for example $\set{1, 2, 5}$. It holds that $\varphi(\set{1,2,5}) = 2$, thus $\set{1}$ is not an \ftc even tough it is a \mc.
Note that we can generalize the example by having $t$ triangles that share a vertex, labeled by $1$. It then holds that $\set{1}$ is a \mc, while $\varphi(\set{1}) = 0$. However, a cut that consists of $1$ and an single vertex from each triangle, is an \ftc whit \aftv $t$.
\end{proof}

\begin{figure}[h]
\centering
\includegraphics[width=8.5cm, scale=0.25]{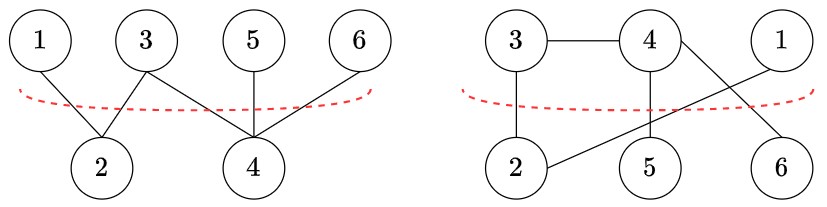}
\caption{An \ftc that is not a \mc}
\label{figure:ftc_no_mc}
\end{figure}

\begin{observation}{There exists a solution for \ftc, which is not a solution for \mc in $G$.}
\end{observation}
\begin{proof}
Let $G$ be a $5$-path, $1,2,3,4,5$, with an additional leaf $6$ connected to vertex $4$ (Figure~\ref{figure:ftc_no_mc}).
It holds that $\set{1,3,4}$ is an \ftc. $C_{\set{1,3,4}} = 4$, while there exists larger cuts, for example $C_{\set{2,4}} = 5$. Thus, $\set{1, 3, 4}$ is not a \mc even tough it is an \ftc.
\end{proof}

\subparagraph*{Greedy moves and stable cuts.} We assume here that we are given an \emph{unweighted} graph $G=(V,E)$. A key observation in our algorithms against an adaptive adversary is that any solution can be transformed into another one where each vertex contributes \emph{many} of its edges to the cut. If a vertex contributes too little, we can just move it to the opposite side of the cut: while this could increase the crossing degree of some vertices (negative contribution to the \aftv), it increases the cut size by more, giving a positive net contribution to the \aftv. We prove this formally in Lemma~\ref{lem:k_greedy_step}, after some formal definitions.

For every $v\in V$ and $S\subseteq V$, let $\sv$ denote the cut obtained from $S$ by switching $v$ to its opposite side, that is, $\sv=S-v$, if $v\in S$, and $\sv=S\cup \{v\}$, otherwise. Given a subset $S\subseteq V$, a constant $k\in \mathbb{N}$, and a vertex $v\in V$, we say that replacing $S$ with $\sv$, i.e., moving $v$ to its opposite side w.r.t. $S$, is a \emph{$k$-greedy step} if $d_S(v) \le (d(v) - k)/2$. A cut $S$ is \emph{$k$-stable} if it has no $k$-greedy step, that is, for every $v\in V$, it holds that $d_S(v) > (d(v) - k)/2$.
For $k=1$, we use \emph{stable} instead of $1$-stable.

\begin{observation}\label{eq:oplusdiff}
For every cut $S$ and a vertex $v$, it holds that $C_{S\oplus v}-C_S=d_{S\oplus v}(v)-d_S(v)$.
\end{observation}

\begin{lemma}\label{lem:k_greedy_step}
Let $v\in V$ be a vertex,  $S\subseteq V$ be a cut, and  $k>0$ be an  integer, such that  $d_S(v) \le (d(v) - k)/2$; then $C_{\sv} \geq C_S + k$, and $\varphi(\sv, k)\geq\varphi(S, k)$.
\end{lemma}

\begin{proof}

Assume, without loss of generality, that $v\in S$ (otherwise, we swap $S$ and $V - S$).  Observation~\ref{eq:oplusdiff} implies that  $C_{\sv} \geq C_S + k$, since $d_S(v) + k \le d(v) - d_S(v) =d_{\sv}(v)$.

For the second claim, we show that for every $F\in {\binom{V}{k}}$, 
$C_{S\sm F, G\sm F} \leq C_{\sv \sm F, G\sm F}$. Assume that $v\notin F$, as otherwise $S-F=\sv-F$, and the claim holds trivially. Recall that $C_{\sv}\geq C_S + k$. In addition, $d_{\sv}(F) \leq d_{S}(F) + k$, since for every $u\in F$, at most one crossing edge is added to the cut (the edge $\set{u,v}$). Putting those together, we have that: $C_{S\sm F, G\sm F} = C_S - d_{S}(F) \leq C_{\sv} - d_{\sv}(F) =  C_{\sv \sm F, G\sm F}$. Since this holds for every $F$, 
we have that $\varphi(S\oplus v, k) \geq \varphi(S, k)$.
\end{proof}

By repeatedly applying a $k$-greedy step to a cut, we keep increasing the cut value, while not decreasing the {\akftv}; thus, after at most $m$ greedy steps, we have a $k$-stable cut with a {\akftv} at least as good as the original one. We let \textsc{StabilizeCut}($G$,$S$,$k$) denote this procedure, which takes as input a graph $G$, a cut $S$ in $G$, and a number $k$, then starting with $S$, repeatedly applies a (arbitrary) $k$-greedy step, while there is one, and returns the obtained $k$-stable cut. The following corollary follows from the reasoning above (the second claim follows by applying \textsc{StabilizeCut} to an optimal {\kftc}).

\begin{corollary}\label{lem:output_of_stable_cut}
Let $S$ be a cut in graph $G=(V,E)$, and $k$ be a positive integer. Let $S' = \textsc{StabilizeCut}(G,S, k)$. 
It holds that $S'$ is $k$-stable, $C_{S'} \ge C_S$ and $\varphi(S', k) \ge \varphi(S, k)$. In particular, every unweighted graph $G=(V,E)$ has a $k$-stable optimal \kftc.
\end{corollary}

\section{Fault Tolerance Against an Adaptive Adversary}\label{sec:Adaptive}

\subsection{A 0.878-Approximation for Multiple Faults}

In this section, we give a $(0.878-\epsilon)$-approximation algorithm for  \kftc on unweighted graphs, for constants $k,\epsilon > 0$. A core tool that we use in our algorithm is an algorithm for the \emph{\smc} problem, where given several graphs defined over the same vertex set, the goal is to find a cut that is large for all graphs simultaneously. A $0.878$-approximation algorithm for this problem with a \emph{constant} number of graphs has been given in~\cite{DBLP:conf/soda/BhangaleKKST18}. The algorithm is based on semidefinite programming techniques.

The main idea behind our algorithm is to separate a constant number of ``heavy'' (high-degree) vertices for which the following holds; given a cut which is large subsequent to any failure of $k$ \emph{heavy} vertices, the cut is large even if light (non-heavy) vertices fail as well. For such a heavy set, a good approximation for \smc on the instances obtained by removing each possibility of $k$ heavy vertices from $G$, should be a good approximation for \kftc on $G$. We give a greedy algorithm that selects the set of heavy vertices. We then consider two cases. We  show that if the heavy vertices do not cover most of the edges in the graph (the ``non-shallow'' case), then   
an approximate solution for \smc with respect to the heavy set gives an approximate solution for \kftc. 
Otherwise (the ``shallow'' case), we identify a set of ``super-heavy'' vertices, which is shown 
to fail in any near-optimal solution. Therefore, finding a near-optimal solution for the original graph reduces to finding a near-optimal solution on the graph remaining by removing the ``super-heavy'' vertices. We show that it can be solved via brutforce, or by finding a good solution to \mc (e.g., obtained via~\cite{DBLP:journals/jacm/GoemansW95}). We prove the following theorem.

\thmkftc*

Before proceeding to the algorithm, we introduce the \smc framework.\label{try32}

\begin{definition}[\smc]
Let $V$ be a vertex set. We are given $k$ edge-weighted graphs, $G_i=(V,E_i)$, $i=1,\dots,k$, on the vertex set $V$, where the weights are normalized, so that $\sum_{e\in E_i}w_e=1$, for each $i$. 
In the (Pareto) \emph{\smc} problem, given the graphs $G_i$ together with thresholds $c_i\in [0,1]$, the goal is to find a cut $S^*\subseteq V$ such that $ C_{S^*,G_i}\ge c_i$, for every $i$. We say that an algorithm is an $\alpha$-approximation algorithm for the problem if for every input $G_i,c_i$, $i=1,\dots,k$, where there exists a cut $S^*$ such that $C_{S^*,G_i}\ge c_i$ for every $i$, the algorithm returns a cut $\sol$ such that $C_{\sol,G_i}\ge \alpha c_i$, for every $i$. 
\end{definition}

\begin{theorem}\cite{DBLP:conf/soda/BhangaleKKST18}\label{thm:simultaneous} For every constant $k\ge 1$ and parameter $n\ge 1$, there is a polynomial-in-$n$ algorithm that computes an $\asmc$-approximate solution to any \smc instance with $k$ weighted graphs on a vertex set of size $n$, in which all non-zero edge-weights are lower-bounded by $\exp(n^{-c})$, for constants $k$ and $c$, and $\asmc = 0.878$
.\end{theorem}

We apply the \smc framework for unweighted graphs $G_i$. We let \textsc{SimultaneousMC} denote the algorithm that gets as input a constant number of unweighted graphs $G_i$, $i=1,\dots,k$, and returns a cut $\sol$ with the following property: for every cut $S^*$ and number $c$ such that $C_{S^*,G_i}\ge c$, for all $i$, it holds that $C_{\sol,G_i}\ge \asmc\cdot c$, for all $i$. This can be achieved by combining the algorithm given in Theorem~\ref{thm:simultaneous} 
(by appropriately scaling  the edge-weights and the thresholds) with a binary search on $c$.

In addition to the \smc algorithm, we use the $\agw$-approximation for \mc due to Goemans and Williamson  \cite{DBLP:journals/jacm/GoemansW95}, for $\agw \approx 0.8786$. We use \textsc{Goemans-Williamson} (with input $G$) to denote this algorithm.
Note that the actual value of the approximation factor $\asmc$ is slightly larger than $0.878$ but is less than $\agw$.

\subparagraph*{The Main Algorithm.}
The inputs to the algorithm (see the pseudocode in Algorithm~\ref{alg:kftc}) are an unweighted graph $G$, and parameters $k$ (number of faults) and $\epsilon$ (precision). First, it computes the set $H$ of heavy vertices via the subroutine \textsc{HeavyVertices}, then applies \textsc{SimultaneousMC} on a collection $\{G_{\neg F} : F\in \binom{H}{k}\}$ of subgraphs containing one subgraph for every failure of $k$ heavy vertices. 
The following notation is used:
for a subset $F\subseteq V$ of vertices, we let $G_{\neg F} = (V, E_{\neg F})$, where $E_{\neg F}=\{e\in E : e\cap F=\emptyset\}$. 
Note that in $G_{\neg F}$, we do not remove the vertices of $F$ from the graph, as opposed to $G-F$, but only the edges adjacent to $F$. 

The pair $(H,\sol)$ is \emph{shallow} if all vertices in $V-H$ have degree at most $3k$, and there are $k$ vertices in $H$ whose removal reduces the weight of $\sol$ below $3k^2/\epsilon$. To state this formally, let us introduce a notation that will be useful later too. 
For a cut $S\subseteq V$, we use $C_{S-k\times H}$ to denote the smallest size of the cut after the failure of any $k$ vertices from $H$, i.e., $C_{S-k\times H}=\min\{C_{S-F,G-F}:F\in \binom{H}{k}\}$. 
Thus, $(H,\sol)$ is shallow if we have $\max_{v\in V-H}d(v)\le 3k$ and $C_{\sol-k\times H}<3k^2/\epsilon$. If $(H,\sol)$ is not shallow, the algorithm simply returns $\sol$. Otherwise, we recompute the cut via \textsc{ShallowFTCut}, using alternative methods.

\begin{algorithm}[t]
\DontPrintSemicolon
\caption{$(\asmc-\epsilon)$-approximation for \kftc}
\label{alg:kftc}

\textbf{Input:} $G=(V,E)$, $k$, $\epsilon$\;
\textbf{Output:} $(\asmc-\epsilon)$-approximation for \kftc\;
$H\la$ \hyperref[alg:heavy_light]{\textsc{HeavyVertices}}($G$, $k$, $\epsilon$)\;
$\sol\gets$ \textsc{SimultaneousMC}($\{G_{\neg F} : F\in \binom{H}{k}\}$)\;
\eIf{$(H,\sol)$ is shallow}
{
    \Return\hyperref[alg:find_remaining_cut]{\textsc{ShallowFTCut}}($G$, $H$, $\sol$, $k$, $\epsilon$)\;
}{
\Return $\sol$
}
\end{algorithm}

The proof of Theorem~\ref{thm:kftc} is split into two parts, addressing shallow and non-shallow cases separately. The running time is dominated by {\smc}.
Before specifying further details, let us mention how the proof follows from the main lemmas addressing those cases.

\begin{proof}[Proof of Theorem~\ref{thm:kftc}]
Let $G$ be a graph and let $S^*$ be an optimal \kftc on $G$. Let $\sol$ be the output of Algorithm~\ref{alg:kftc} on $G$, $k$, $\epsilon$.  We show that $\varphi(\sol, k) \geq (\asmc-\epsilon)\cdot \varphi(S^*, k)$. Lemma~\ref{lem:good_cut_first_case} provides this for the non-shallow case, while Lemma~\ref{lem:good_cut_second_case} provides it in the shallow case. The algorithm is indeed polynomial, since the sub-routines are such, and the input to \textsc{SimultaneousMC} consists of ${\binom{|H|}{k}}=O(k/\epsilon)^k=O(1)$ subsets, where $|H|=O(k^2/\epsilon)$ is proven in Lemma~\ref{lem:heavysmall}.
\end{proof}

\begin{algorithm}

\DontPrintSemicolon
\caption{\textsc{HeavyVertices}}
\label{alg:heavy_light}

\textbf{Input:} $G=(V,E)$, $k$, $\epsilon$\;
\textbf{Output:} $H\subseteq V$, the set of heavy vertices\;
Let $v_1, \dots, v_n$ be an ordering of vertices by non-increasing degree \;

$\sigma \la 0$, $i\la 1$, $H\la \set{v_1,\dots, v_k}$ \;
\While{$d(v_{k+i}) > (\epsilon\cdot \asmc/k)\cdot \sigma$ \textbf{\emph{and}} $d(v_{k +i}) > 3k$ \label{line:heavy_light_cond}}{
    $\sigma \la \sigma + (d(v_{k+i})-3k)/4$\;
    $H \la H \cup \set{v_{k+i}}$\;
    $i \la i + 1$\;
}
\Return $H$\;
\end{algorithm}

The selection of heavy vertices (Algorithm~\ref{alg:heavy_light}) is done by a simple greedy procedure, where we sequentially select vertices in the heavy set $H$ in a non-increasing order by degree. The selection stops either when the remaining vertices ($V- H$) have a small degree (at most $3k$) or when $H$ has sufficiently many incident edges (used in Lemma~\ref{lem:good_cut_first_case}).  By Corollary~\ref{lem:output_of_stable_cut}, any cut can be transformed into one with a similar {\akftv}, where every vertex $v$ has crossing degree at least $(d(v)-k)/2$, and at least $(d(v)-3k)/2$, after $k$ failures. Thus, heavy vertices are guaranteed to contribute $\sigma$ in the ``stable version'' of every cut. The degree constraint ensures that we do not select vertices that are unnecessary, according to this logic, which helps us keep the size of $H$ bounded. 

In the analysis below, we often use the notation $\sigma_i$ to denote the value of $\sigma$ after the $i$-th iteration.

\begin{restatable}{lemma}{lemheavysmall}\label{lem:heavysmall}
Algorithm~\ref{alg:heavy_light} terminates within $t=4(3k^2 + k)/(\epsilon \cdot \asmc)$ iterations. In particular, $|H|\le t + k$.
\end{restatable}

\begin{proof}
If $d(v_{k+t}) \leq 3k$, then by the condition in Line~\ref{line:heavy_light_cond}, the algorithm terminates before the $t$-th iteration; therefore, assume $d(v_{k+t})> 3k$. 
For every $i\leq t$, after the $i$-th iteration, it holds that $\sigma_i = \sum_{j=1}^i {(d(v_{k+j}) - 3k)/4}$;
thus, after $t$ iterations, 
\begin{align*}
\sigma_{t} &= \sum_{j=1}^{t} {\frac{d(v_{k+j}) - 3k}{4}} \geq t \cdot \frac {d(v_{k+t}) - 3k}{4} =   \frac{3k^2 + k}{\epsilon \cdot \asmc} \cdot \left(d(v_{k+t}) - 3k\right)\\
&=\frac{k}{\epsilon \cdot \asmc}\cdot d(v_{k+t}) + \frac{3k^2}{\epsilon \cdot \asmc} \cdot d(v_{k+t}) - \frac{3k^2(3k+1)}{\epsilon \cdot \asmc}\geq \frac{k}{\epsilon \cdot \asmc} d(v_{k+t})
\ ,
\end{align*}
where in the first inequality, we use the fact that the vertices are processed in a non-increasing order of degrees, and in the last inequality, we use the assumption that $d(v_{k+t})\ge 3k+1$.
It follows that $d(v_{k+t}) \leq (\epsilon \asmc/k)\cdot \sigma_{t}$, and using $d(v_{k+t+1}) \le d(v_{k+t})$, we get that the algorithm terminates within the first $t$ iterations, by the condition in Line~\ref{line:heavy_light_cond}.
\end{proof}

\subsubsection{The Non-shallow Case}\label{append:kftc_non_shallow}

Recall that in the non-shallow case, the cut $\sol$ and the set $H$ of heavy vertices are such that \emph{either $d_{max}=\max_{v\in V- H}d(v)> 3k$ or $C_{\sol-k\times H}\ge 3k^2/\epsilon$ holds}. 
Let $S^*_{smc}$ be an  optimal solution of Simultaneous \mc for the instances $\set{G_{\neg F}:F\in \binom{H}{k}}$. Let
$S^*_{ft}$ be an optimal solution for \kftc on $G$.

Let us begin with an observation connecting the three cuts $\sol,S^*_{smc}$, and $S^*_{ft}$.\label{tryB1}   
\begin{observation}\label{obs:p_star_greater_than_p}

It holds that $\ptilde\ge \asmc\cdot \pstar\ge \asmc\varphi(S^*_{ft},k)$.
\end{observation}
\begin{proof}
The first inequality holds because $\sol$ is an $\asmc$-approximation to the \smc problem as described in Theorem~\ref{thm:simultaneous}, while the second one holds since by definition, $S^*_{smc}$ is the cut $S$ optimizing $C_{S-k\times H}$, and $S^*_{ft}$ is the one optimizing $C_{S-k\times V}$.
\end{proof}

 We  also use the following lower bound on $\pstar$ in terms of node degrees, in order to show that the degree of light (non-heavy) vertices is small in comparison with the cut size even after a heavy failure, implying that light vertex faults can be tolerated. 
\begin{lemma}\label{lem:p_star_is_large}
$\pstar \geq \sum_{i=1}^{n-k}{(d(v_{k+i}) - 3k)_+/4}$ where $v_1,\dots, v_n$ are sorted by  degree, in descending order, and $(x)_+ = \max\{x, 0\}$ for any argument $x$.
\end{lemma}

\begin{proof}
By Corollary~\ref{lem:output_of_stable_cut}, there is an optimal solution $S\subseteq V$ for \kftc that satisfies $d_S(v) > (d(v)-k)/2$, for every $v\in V$. Let $F\in \binom{H}{k}$. Since every vertex $v\notin F$ has at most $k$ neighbors in $F$, we have, in $G-F$, that $d_{S-F,G-F}(v)> (d(v)-k)/2-k = (d(v)-3k)/2$, and also $d_{S-F,G-F}(v) \ge 0$. Hence, we have
\[
C_{S\sm F, G\sm F}=\sum_{v\in V-F} d_{S-F,G-F}(v)/2\ge \sum_{v\in V-F} (d(v)-3k)_+/4 \ge \sum_{i = 1}^{n-k}{(d(v_{k+i}) - 3k)_+/4}\ ,
\]
where the last inequality holds by the assumption on the ordering of vertices (and since $|V-F|=n-k$). Since $F$ is an arbitrary $k$-subset of $H$, we conclude that $C_{S-k\times H}\ge \sum_{i=1}^{n-k} (d(v_{k+i})-3k)_+/4$, and the claim now follows from  $\pstar\ge C_{S-k\times H}$ (by the definition of $S^*_{smc}$).
\end{proof}

We are now ready to prove that in the non-shallow case, $\sol$ is a $(\asmc-\epsilon)$-approximation for \kftc.

\begin{restatable}{lemma}{lemgoodcutfirstcase} \label{lem:good_cut_first_case}
If $(H,\sol)$ is not shallow, then it holds that
$\varphi(\sol, k)\geq (\asmc-\epsilon)  \varphi(S^*_{ft},k)$, for an optimal \kftc $S^*_{ft}$.
\end{restatable}

\begin{proof} Note that in this case, we have either $d_{max}=\max_{v\in V-H} d(v) > 3k$ or $\ptilde \geq 3k^2/\epsilon$.
Consider an arbitrary $F\in {\binom{V}{k}}$. It suffices to show that
$C_{\sol\sm F, G\sm F} \geq (1-\epsilon)\asmc\cdot \varphi(S^*_{ft}, k)$. 
Let $H'=F\cap H$ and $L'=F\cap (V-H)$ be the heavy and light (non-heavy) vertices in $F$, respectively.  Observe that $C_{\sol-F,G-F}$ is obtained from $C_{\sol-H',G-H'}$ by removing the set $L'$ of at most $k$ light vertices. Since each vertex in $L'$ has degree at most $d_{max}$,
\begin{equation}\label{eq:stilde0}
C_{\sol-F,G-F}\ge C_{\sol-H',G-H'}-kd_{max}\ .
\end{equation}
On the other hand, we have, from Observation~\ref{obs:p_star_greater_than_p} (and using $H'\subseteq H$, $|H'|\le k$) that
\begin{equation}\label{eq:stilde01}
C_{\sol-H',G-H'}\ge C_{\sol-k\times H}\ge \asmc\cdot \pstar\ge \asmc\varphi(S^*_{ft}, k)\ .
\end{equation}
In the following, we show that $d_{max}\le (\epsilon/k)C_{\sol-H',G-H'}$, which implies the claim by combining (\ref{eq:stilde0}) and (\ref{eq:stilde01}).

If $d_{max}\le 3k$ and $\ptilde\ge 3k^2/\epsilon$, then the claim holds, since $C_{\sol-H',G-H'}\ge \ptilde\ge (k/\epsilon)d_{max}$. Thus, we may henceforth focus on the case $d_{max}>3k$.
Using (\ref{eq:stilde01}) and Lemma~\ref{lem:p_star_is_large}, we have
\begin{equation}\label{eq:stilde1}
C_{\sol-H',G-H'}\ge  \asmc\cdot \pstar \geq \asmc \cdot \sum_{i=1}^{n-k}{(d(v_{k+i}) - 3k)_+/4}\ .
\end{equation}
Let $t$ be the number of iterations after which Algorithm~\ref{alg:heavy_light} terminates. Since $d_{max} > 3k$, it follows from the algorithm description that the vertex $v_{k+t+1}$ satisfies $d(v_{k+t+1}) \le (\epsilon\asmc/k)\sigma_{t}$ and that $d_{max}=d(v_{k+t+1})$. Hence $d_{max}\le (\epsilon\asmc/k)\sigma_t$. On the other hand, we have
\begin{equation}\label{eq:stilde2}
\sigma_t=\sum_{i=1}^{t}{(d(v_{k+i}) - 3k)/4}\le \sum_{i=1}^{n-k}{(d(v_{k+i}) - 3k)_+/4}\le (1/\asmc)C_{\sol-H',G-H'}\ ,
\end{equation}
where we used (\ref{eq:stilde1}). This gives us the bound $d_{max}\le (\epsilon\asmc/k)\sigma_t\le (\epsilon/k)C_{\sol-H',G-H'}$, as claimed, which completes the proof. 
\end{proof}

\subsubsection{The Shallow Case}\label{append:kftc_shallow}

Recall that $\X=\set{v\in V | \space (d(v)-3k)/2 > \ptilde/\asmc}$, $G_R=G-\X$ is the graph obtained by removing the super-heavy vertices. We let $m_R$ be the number of edges in $G_R$,  $n_R$ be the number of non-isolated vertices in $G_R$, and  $\md=6k^2/(\asmc \epsilon)+3k$ be a parameter. We show that the algorithm \textsc{ShallowFTCut}, as described in Algorithm~\ref{alg:find_remaining_cut}, returns a cut that in the shallow case is a $(\asmc-\epsilon)$-approximation.\label{tryB3}

\begin{algorithm}
\DontPrintSemicolon
\caption{\textsc{ShallowFTCut}}
\label{alg:find_remaining_cut}

\textbf{Input:} $G=(V,E)$, $H$, $\sol$, $k$, $\epsilon$\;
\textbf{Output:} Cut $\hat S \subseteq V$\;
\eIf{$m_R < 2k\md/(\asmc\epsilon)$}{
\For{every $S'\subseteq V_R$ ($V_R$ is of constant size)} {
 Compute $\varphi(S', k-|\X|)$ \;
}
$\hat S\gets \arg\max_{S'\subseteq V_R }\varphi(S', k-|\X|)$\;
}{
$\hat S\gets $ \textsc{Goemans-Williamson}$(G_R)$\;
}
\While{$\exists v\in \X$ such that $d_{\hat S}(v) \le (d(v) - k)/2$ }{
    $\hat S\la \hat S \oplus v$\;
}
\Return $\hat S$\;
\end{algorithm}

Let us begin with two observations, which show that $\X$ is indeed small  and is contained in every worst-case fault set in stable cuts, and that the vertices outside $\X$ have small degree, bounded by $\ell$.
\begin{lemma}\label{lem:X_failure}
Let $S\subseteq V$ be a cut such that $d_S(v) \geq (d(v)-k)/2$, for every $v\in \X$. If $F\in {\binom{V}{k}}$ is such that $C_{S\sm F, G\sm F} = \varphi(S, k)$, then $\X\subseteq F$. In particular, $|\X|\le k$.
\end{lemma}

\begin{proof}
Let $v\in \X$, and assume, towards a contradiction, that $v\notin F$. We have $d_S(v)\ge (d(v)-k)/2$, and since $|F|\le k$, 
 $d_{S-F,G-F}(v)\ge (d(v)-k)/2-k=(d(v)-3k)/2$. Since $v\notin F$, $C_{S-F,G-F}\ge d_{S-F,G-F}(v)$; hence, 
\[
C_{S-F,G-F}\ge(d(v)-3k)/2 > \ptilde/\asmc\ge \pstar\ge \varphi(S^*_{ft}, k)\ .
\]
where we use $v\in \X$ in the second inequality, and Observation~\ref{obs:p_star_greater_than_p} in the last two inequalities. This implies that $\varphi(S,k)=C_{S-F,G-F}> \varphi(S^*_{ft},k)$, which is a contradiction to the definition of $S^*_{ft}$.
\end{proof}

\begin{observation}\label{obs:degreeboundell}
$d(v) \leq \md$ holds for all $v\in V- \X$.
\end{observation}

\begin{proof}
For every $v\notin \X$,  it holds that $(d(v)-3k)/2 \leq \ptilde/\asmc$, which implies that $d(v) \leq 2\ptilde/\asmc + 3k<6k^2/(\asmc \epsilon)+3k=\md$, since, by our assumption, $\ptilde < 3k^2/\epsilon$.
\end{proof}

Let $\hat S$ be the output of Algorithm~\ref{alg:find_remaining_cut}. First, we show that $\hat S-\X$ is a good solution in $G_R$. Then, we prove the main claim of this subsection, that is, that $\hat S$ is a $(\asmc-\epsilon)$-approximation for \kftc in the shallow case.\label{tryB5}

\begin{lemma}\label{lem:largercut}
If $d_{max} \leq 3k$ and $\ptilde < 3k^2/\epsilon$, then $\hat S-\X$ is a $(1-\epsilon)\asmc$-approximation for \xftc{(k-|\X|)} in $G_R$.
\end{lemma}

\begin{proof}
Indeed, in the case $m_R<2k\md/(\asmc\epsilon)$, $\hat S - \X$ is an optimal solution for \xftc{(k-|\X|)} by the description of the algorithm, and the claim follows. If $m_R\ge 2k\md/(\asmc\epsilon)$, note that $C_{\hat S-\X,G-\X}\ge \agw (m_R/2)\ge k\md/\epsilon$, since $\hat S-\X$ is an $\agw$-approximation for \mc in $G_R$. By Observation~\ref{obs:degreeboundell}, $d(v)\le \md$ holds for each vertex in $V- \X$. If $k - |\X|$ vertices fail, the cut size is still at least $C_{\hat S-\X,G-\X}-k\md\ge (1-\epsilon)C_{\hat S-\X,G-\X}$. The claim then follows from the fact that the optimal \aftv is bounded by the optimal \mc size, $\hat S-\X$ is an $\agw$-approximation for \mc in $G_R$, and $\agw\ge\asmc$
\end{proof} 

\begin{restatable}{lemma}{lemgoodcutsecondcase}\label{lem:good_cut_second_case}
If $(H,\sol)$ is shallow, then it holds that $\varphi(\hat S, k) \geq (\asmc-\epsilon)  \cdot \varphi(S^*_{ft}, k)$, for an optimal \kftc $S^*_{ft}$.
\end{restatable}
\begin{proof}
Note that in this case, $d_{max} \leq 3k$, and $\ptilde < 3k^2/\epsilon$.
Let $S^*\subseteq V$ be an optimal solution for \kftc on $G$ such that for every $v\in V$, $d_S(v) \geq (d(v)-k)/2$; such cut exists, by Corollary~\ref{lem:output_of_stable_cut}. Note that $\hat S$ satisfies $d_{\hat S}(v)\ge (d(v)-k)/2$, for every $v\in \X$.  
By Lemma~\ref{lem:X_failure}, it holds that $\X$ belongs to \emph{every} worst-case failure set for those cuts, that is, if $F^*\in {\binom{V}{k}}$ is such that $C_{S^*\sm F^*, G\sm F^*} = \varphi(S^*, k)$, then $\X\subseteq F^*$, and similarly, if $\hat F\in {\binom{V}{k}}$ is such that $C_{\hat S\sm \hat F, G\sm \hat F} = \varphi(\hat S, k)$, then $\X\subseteq \hat F$. We can conclude that $\varphi(S^*-\X,k-|\X|,G_R)=\varphi(S^*,k)$, and similarly, $\varphi(\hat S-\X,k-|\X|,G_R)=\varphi(\hat S,k)$. Lemma~\ref{lem:largercut} shows that $\varphi(\hat S-\X,k-|\X|,G_R)\ge (1-\epsilon)\asmc\cdot \varphi(S^*-\X,k-|\X|,G_R)$. Combining these together, we see that $\varphi(\hat S,k)\ge (1-\epsilon)\asmc\cdot  \varphi(S^*,k)=(1-\epsilon)\asmc\cdot \varphi(S^*_{ft},k)\ge (\asmc-\epsilon)  \cdot \varphi(S^*_{ft}, k)$.
\end{proof}

\subsection{A Combinatorial 1/2-Approximation for a Single Fault}\label{sec:CombLS}

In the case of a single fault, we have the following result, that is, a simple and efficient $1/2$-approximation for the case of a single fault. Moreover, we show that an \aftv of $(m-\Delta)/2$ can be achieved, for $\Delta\ge  3$, while $m-\Delta$ is an (easy) upper bound. 
\greedyftc*

\subparagraph{The Challenge.} In the discussion below, we call a vertex $v$ \emph{critical} for a cut $S$ if $C_{S-v,G-v}=\varphi(S)$.

It is well-known (and easy to show) that every stable cut is a $1/2$-approximate \mc. This even holds for \ftc, with $\Delta=2$ (see Lemma~\ref{lem:specific_case_half}). However, in general, while we know that greedy steps (moving a vertex $v$ with $d(v) < d_S(v)/2$) never decrease the \aftv (Lemma~\ref{lem:k_greedy_step}), a stable cut can be a poor approximation for \ftc.
Consider, for example, a graph that consists of $t$ triangles with a single common vertex $u$. Note that $d(u) = \Delta= 2t$, $d(v) = 2$, for every $v\neq u$, and $m = 3t$. The cut $S'=\set{u}$ is a stable  cut, with  $\varphi(S') = 0$. In order to transform $S'$ into a $1/2$-approximation, we have to decrease the crossing degree of the critical vertex $u$ without decreasing the size of the cut. This can be done by moving a neighbor $v$ of $u$ from the opposite side of the cut, since $d_{S'}(v) = d(v)/2$.
 
In general, moving such vertex $v$ (which we call a \emph{neutral} move below) does not change the size of the cut, and  decreases the crossing degree of $u$. Nevertheless, it does not always imply that the \aftv increases, as there can be an additional critical vertex $u'$ in $S$ that is not affected, or that moving $v$ creates a new critical vertex $u''$ with the same crossing degree as $u$. 
 
Our algorithm is based on some key structural properties of stable cuts that we prove.
Essentially, we show that any given cut $S$ with {\aftv} less than $(m-\Delta)/2$ either admits a \emph{greedy move}, or a \emph{neutral move followed by a greedy move}, or a \emph{neutral move that increases the {\aftv}} (see Lemma~\ref{lem:has_greedy_step}). Our algorithm is then a repeated application of such steps until the cut has the desired {\aftv}; thus, it can be seen as a local search over two-move combinations, for maximizing the \emph{sum} of the cut size and {\aftv}.

Our key technical observation is that in a balanced cut $S$ with an {\aftv} less than $(m-\Delta)/2$, the critical vertex is unique. Moreover, letting $x_S(v)=d_S(v)-d(v)/2$ denote the \emph{excess} contribution of  a vertex $v$ to the cut, it holds for the critical vertex $u$ that $x_S(u)>\sum_{v\neq u} x_S(v) + \Delta-d(u)$ (see Lemma~\ref{lem:critical_xS_ineq}). Note that in a stable cut $S$, $x_S(v)$ is a non-negative multiple of $1/2$, for all $v$. In most typical cases (e.g., when $d(u)<\Delta$, or when there are not too few nodes $v$ with $x_S(v)>0$), the inequality above quickly gives us the properties we claimed. However, covering all cases turns out to be quite tedious (see Lemma~\ref{lem:has_greedy_step}).

\subparagraph*{Outline of the Algorithm.} We give the pseudocode of the algorithm in Algorithm~\ref{alg:Greedy_ftc}. If $\Delta \le 2$ then the algorithm returns an arbitrary stable cut. For $\Delta > 2$, the algorithm initializes a solution $\sol$ to be the empty set, and then updates it in iterations, until $\varphi(\sol) \geq (m-\Delta)/2$. In each iteration, the algorithm chooses a vertex $v$ and moves it to the other side of the cut, as follows.
First, if there is a vertex $v$ such that $\varphi(\sol\oplus v) \geq (m-\Delta)/2$ then the algorithm moves $v$. We call this a \emph{type-$0$} step. Note that after applying a type-$0$ step, the algorithm terminates. Otherwise, if there is a vertex $v$ with $d_{\sol}(v) < d(v)/2$, then the algorithm moves it to the other side of the cut. This step is called a \emph{type-$1$ step}. Otherwise, if there is a vertex $v$ with $d_{\sol}(v) = d(v)/2$ such that $\varphi(\sol\oplus v)\ge\varphi(\sol)$ and a type-$1$ step can be applied to $\sol\oplus v$, then the algorithm moves $v$ to the other side of the cut. This is called a \emph{build-up step}. Finally, if none of the above conditions hold, the algorithm takes a vertex $v$ with $d_{\sol}(v) = d(v)/2$ that satisfies $\varphi(\sol\oplus v) > \varphi(\sol)$, and moves $v$ to the other side of $\sol$. We prove that in this case, such a vertex exists, and hence this covers all possibilities. The latter step is called a \emph{type-$2$ step}. 

\begin{algorithm}[H]
\DontPrintSemicolon
\caption{Combinatorial $1/2$-approximation for \ftc}
\label{alg:Greedy_ftc}

\textbf{Input:} $G=(V,E)$ \;
\lIf{$\Delta\le 2$}{\Return {\textsc{StabilizeCut}}$(G,\emptyset, 1)$}
$\sol\la \emptyset$\;
\While{$\varphi(\sol) < (m-\Delta)/2$ }{
    
    \If{$\exists v, \varphi(\sol\oplus v) \geq (m-\Delta)/2$}{
        $\sol\la\sol\oplus v$ \tcp*[f]
        {type-$0$ step}}
    \ElseIf{$\exists v,d_{\sol}(v)<d(v)/2$}{
    $\sol \la \sol\oplus v$ \tcp*[f]{type-$1$ step}}
    \ElseIf{$\exists v,w,\left(d_{\sol}(v)=d(v)/2 \text{ \textbf{\emph{and}} } \varphi(\sol\oplus v)\ge\varphi(\sol) \text{ \textbf{\emph{and}} }  d_{\sol\oplus v}(w)<d(w)/2\right)$}{
    $\sol \la \sol\oplus v$ \tcp*[f]{build-up for another type-$1$ step}}
    \Else{
    $v\gets $ a vertex such that $d_{\sol}(v)=d(v)/2$ and $\varphi(\sol\oplus v)>\varphi(\sol)$\;
    $\sol\gets \sol\oplus v$ \tcp*[f]{type-$2$ step}}
    }
\Return $\sol$\;
\end{algorithm}

\subparagraph*{Outline of the Proof.} The approximation is based on the a simple observation, that the optimal \aftv is bounded by $m-\Delta$, which holds since after failing a degree-$\Delta$ vertex, only $m-\Delta$ edges remain in the graph. Thus, in the proof of Theorem~\ref{thm:Greedy_ftc}, our aim is to get a cut with \aftv $(m-\Delta)/2$.   
If $\Delta \leq 2$, this is not always achievable (consider, e.g., a triangle). We show that nevertheless, any stable cut is a $1/2$-approximation (see Lemma~\ref{lem:specific_case_half}). 
If $\Delta > 2$, we show that in every two consecutive iterations, either the size of the cut increases or the \aftv of the cut increases, while both never decrease. Since $\varphi(\sol)$ and $C_{\sol}$ are bounded, we get that the algorithm terminates. By the pseudocode of the algorithm it follows that the algorithm terminates only when $\varphi(\sol)\geq (m-\Delta)/2$, i.e., $\sol$ is a $1/2$-approximation.

The approximation is based on the following simple observation, which holds since after failing a degree-$\Delta$ vertex, only $m-\Delta$ edges remain in the graph.\label{tryC1}
\begin{observation}\label{obs:greedy:max_ft_size}
Let $S^*$ be an optimal  \ftc in a graph $G$. It holds that $\varphi(S^*) \leq m-\Delta$.
\end{observation}

We use the following notation.
\begin{definition}[excess]
The \emph{excess} of a vertex $v$ in a cut $S\subseteq V$ is $x_S(v) = d_S(v) - d(v)/2$.
\end{definition}
Note that for a stable cut $S$ (see Section~\ref{sec:prelim}), $x_{S}(v) \geq 0$ for every $v\in V$. In addition, if $d(v)$ is even, then $x_{S}(v)$ is an integer. Otherwise, $x_{S}(v) = a + 1/2$ for some integer $a$.
Now, we prove the properties that are required for showing the correctness of our algorithm.

\begin{lemma}\label{lem:critical_xS_ineq}
Let $S$ be a stable cut in a graph $G=(V,E)$ such that $\varphi(S)<(m-\Delta)/2$. Then $S$ has a unique critical vertex vertex $u$, and $u$ satisfies
\begin{equation}\label{eq:critical_xS_ineq}
    d_{S}(u) > \underset{v \neq u}{\sum} x_{S}(v) + \Delta-d(u)/2\ .
\end{equation}
Moreover, $u$ has a neighbor $w$ in its opposite side of the cut, which satisfies $x_S(w)=0$. \label{tryC3}
\end{lemma}

\begin{proof}
First, we show that $S$ has a unique critical vertex.
Since for every vertex $v$, $d_{S}(v) = d(v)/2 + x_{S}(v)$ and $\sum_{v\in V}{d(v)/2} = m$, we get that 
\begin{equation}\label{eq:cutviaexcess}
C_{S} = \frac{1}{2}\underset{v\in V}{\sum}{d_{S}(v)} = \frac{1}{2}\underset{v\in V}{\sum}{\left(\frac{d(v)}{2} + x_{S}(v)\right)} = \frac{m}{2} + \underset{v\in V}{\sum}{\frac{x_{S}(v)}{2}}\ .
\end{equation}
Let $u$ be a critical vertex, and assume, without loss of generality,  that $u\in S$ (otherwise we swap $S$ and $V- S$). On one hand, we have  $\varphi(S) = C_{S} - d_{S}(u)=m/2+\sum_{v\in V}x_S(v)/2 - d_S(u)$, and on the other hand, we have $\varphi(S) < (m-\Delta)/2$, which together imply: 
\[
m/2 + \underset{v\in V}{\sum} x_{S}(v)/2 - d_{S}(u) < (m-\Delta)/2\ .
\]
After a rearrangement, the latter implies \eqref{eq:critical_xS_ineq}.
Using $d_S(u)=d(u)/2+x_S(u)$ in \eqref{eq:critical_xS_ineq} and simplifying, we get $x_{S}(u) > \underset{v \neq u}{\sum} x_{S}(v) + \Delta-d(u)\ge \underset{v \neq u}{\sum} x_{S}(v)$.  Since $u$ is an arbitrary critical vertex, this implies that  $u$ is the only critical vertex of $S$. 

Next, let us show that there is a neighbor $w\in V- S$ of $u$ (recall that $u\in S$)  with $x_{S}(w) = 0$. Assume to the contrary that for every $v\notin S$  such that $\{u,v\}\in E$, it holds that $x_{S}(v) \geq 1/2$ (recall that $S$ is stable, and hence $x_{S}(v)$ is a non-negative integer multiple of $1/2$). Using (\ref{eq:cutviaexcess}), this implies:
\begin{align*}C_{S} = \frac{m}{2} + \underset{v\in V}{\sum}\frac{x_{S}(v)}{2} &\geq \frac{m}{2} + \frac{1}{2}\left(x_{S}(u)+\underset{v: \{u, v\}\in \delta(S)}{\sum}x_{S}(v)\right)\\
&\ge \frac{m}{2}+\frac{1}{2}\left(x_S(u)+\frac{d_S(u)}{2}\right)
\ge \frac{m}{2} + x_S(u)\ ,
\end{align*}
where we use $d_S(u)=|\set{ v: \set{u,v}\in \delta(S)}|$ in the second inequality, and $d_S(u)=d(u)/2+x_S(u)\ge 2x_S(u)$, in the third one. Since $u$ is the critical vertex of $S$, this gives that\[\varphi(S) = C_S - d_{S}(u) \geq \frac{m}{2} + x_{S}(u) - d_{S}(u) = \frac{m}{2} - \frac{d(u)}{2} \geq (m-\Delta)/2\ ,\] in contradiction to $\varphi(S) < (m-\Delta)/2$. This completes the proof. 
\end{proof}

\subparagraph*{}We can now show that in the case of $\Delta \leq 2$, any stable cut is a $1/2$-approximation.\label{tryC4}

\begin{lemma}\label{lem:specific_case_half}
Every stable cut $S$ in a graph $G=(V,E)$ with $\Delta\le 2$ is a $1/2$-approximation for \ftc. 
\end{lemma}

\begin{proof}
We assume w.l.o.g. that $G$ is connected, i.e., it is either a path-graph or a cycle-graph.
If $\varphi(S)\ge (m-\Delta)/2$, then Observation~\ref{obs:greedy:max_ft_size} implies that  $S$ is a $1/2$-approximation. Otherwise, the conditions of Lemma~\ref{lem:critical_xS_ineq} apply, and hence there is a unique critical vertex vertex $u$, which satisfies $x_S(u) > \sum_{v\neq u}{x_S(v)}$ (where we use that $d_S(u)=x_S(u)+d(u)/2$). 
Since $\Delta = 2$, we have  $x_S(v) \le 1$, for every vertex $v$, and hence $x_S(v) \leq 1/2$, for every $v\neq u$. It cannot be that there is only one vertex $v$ with $x_S(v) = 1/2$, as this would imply that there is only one (odd number) vertex with an odd degree. In addition, if there are two vertices $v,v'\in V - u$ such that $x_S(v)=x_S(v')=1/2$, we get a contradiction to $x_S(u) > \sum_{v\neq u}{x_S(v)}$. Therefore $x_S(u) = 1$, and $x_S(v) = 0$, for every $v\neq u$, which implies that all degrees are even, i.e.,
 $G$ is a cycle graph. The uniqueness of the critical vertex also implies that $d_S(u)=2$ and $d_S(v)=1$, for all $v\neq u$; therefore, we have $C_S=(2+n-1)/2=(n+1)/2$,  $\varphi(S)=C_S - d_S(u)=(n-3)/2$, and $n$ is odd. Note that for every cut $S'$ of an odd cycle, there is at least one edge that does not cross the cut, which means that $\varphi(S')\le m-\Delta-1$. Since  $\varphi(S)=(n-3)/2=(m-3)/2=(m-\Delta-1)/2$, we have that $S$ is a $1/2$-approximation.
\end{proof}

In the following lemma, we show that given a cut $S$ with a unique critical vertex $u$, it holds that $\varphi(S\oplus w) \geq \varphi(S)$ for every vertex $w$ with $x_S(w) = 0$ which is not in the same side of the cut as $u$. Therefore, if the algorithm cannot apply a build-up step, it means that for every $w$ as described above, $S\oplus w$ is a stable cut.

\begin{lemma}\label{lem:opluscut}
Let $S$ be a stable cut in a graph $G=(V,E)$ with a unique critical vertex $u\in S$. Let $w\notin S$ with $x_S(w)=0$. Then, $\varphi(S\oplus w)\ge \varphi(S)$. 
\end{lemma}

\begin{proof}
By definition, it holds that
\begin{equation}\label{eq:opluscut}
\varphi(S\oplus w) = C_{S\oplus w} - \max_{v\in V}d_{S\oplus w}(v)=C_{S} - \max_{v\in V}d_{S\oplus w}(v)\ ,\end{equation}
where we use $C_{S\oplus w} = C_{S} - d_{S}(w) + d_{S\oplus w}(w)$ (Observation~\ref{eq:oplusdiff}), and $d_{S}(w) = d_{S\oplus w}(w)+2x_S(w)=d_{S\oplus w}$. We have that $d_{S\oplus w}(u) \leq d_{S}(u)$, since if $\{u,w\}\in E$ then $d_{S\oplus w}(u) = d_S(u) - 1$, and otherwise $d_{S\oplus w}(u) = d_S(u)$. For every other vertex $v\neq u$, we have that $d_{S\oplus w}(v)\le d_S(v)+1\le d_S(u)$, where the last inequality holds since $u$ is the unique critical vertex of $S$, i.e., $d_S(v)<d_S(u)$. 
Altogether, we have that $\max_{v\in V}d_{S\oplus w}(v)\le d_S(u)$, which, together with \eqref{eq:opluscut} implies that $\varphi(S\oplus w)\ge C_S-d_S(u)=\varphi(S)$, as claimed.
\end{proof}
\label{tryC5}

Now, we show that if the conditions for the type-$0$, type-$1$ and build-up steps do not hold, then Algorithm \ref{alg:Greedy_ftc} can apply the type-$2$ step. It holds that $\varphi(\sol) < (m-\Delta)/2$ because otherwise the algorithm terminates. Since the algorithm does not apply a type-$1$ step, we get that $\sol$ is a stable cut and Lemma~\ref{lem:critical_xS_ineq} implies that $\sol$ has a unique critical vertex $u$. In addition, for every vertex $w$ which is not in the same side of the cut as $u$ with $x_{\sol}(w) = 0$, it holds that $\sol\oplus w$ is a stable cut that satisfies $\varphi(\sol\oplus w) < (m-\Delta)/2$ (otherwise the algorithm can apply a type-$0$ or a build-up step). We show that there is a neighbor $v$ of $u$ with $x_{\sol}(v) = 0$, such that moving $v$ to the other side increases the \aftv of $\sol$, therefore Algorithm \ref{alg:Greedy_ftc} can apply a type-$2$ step by choosing $v$.

\begin{lemma}\label{lem:has_greedy_step}
Let $S$ be a stable cut in a graph $G=(V,E)$ with $\Delta>2$, such that  $\varphi(S)<(m-\Delta)/2$. Denote the unique critical vertex of $S$ by $u$, assume w.l.o.g. that $u\in S$ and assume that for every vertex $w\notin S$ with $x_S(w)=0$, $S\oplus w$ is a stable cut and satisfies $\varphi(S\oplus w)<(m-\Delta)/2$. Then there is a vertex $v$ with $x_S(v)=0$ such that $\varphi(S\oplus v)>\varphi(S)$. 
\end{lemma}

\begin{proof}
Let us choose a vertex $v\in N(u)\setminus S$ with $x_S(v)=0$ (where $N(u)$ is the neighborhood of $u$), as follows: if there is a vertex $w\in V$ such that $x_S(u),x_S(w)>0$, while $x_S(w')=0$, for all $w'\notin \{u,w\}$, and there is a vertex $w'\in (N(u)\setminus N(w))\setminus S$ ($w'$ is a neighbor of $u$ on its opposite side, and a non-neighbor of $w$), then let $v=w'$ (note that $x_S(v)=0$). Otherwise, pick an arbitrary $v\in N(u)\setminus S$ with $x_S(v)=0$ (as provided by Lemma~\ref{lem:critical_xS_ineq}). Let $u'$ be the unique critical vertex in $S\oplus v$ (by Lemma~\ref{lem:critical_xS_ineq}). We claim that $u'=u$. Before proving the claim, let us see how it implies the lemma. By Observation~\ref{eq:oplusdiff}, $C_{S\oplus v}\ge C_S$, and $d_{S\oplus v}(u) = d_S(u)-1$, thus $u' = u$ implies that $\varphi(S\oplus v)=C_{S\oplus v}-d_{S\oplus v}(u)>C_{S}-d_{S}(u)=\varphi(S)$, and the lemma follows.

Assume henceforth, towards a contradiction, that $u'\neq u$.
By Lemma~\ref{lem:critical_xS_ineq}, we have 
\begin{equation}\label{eq:localbottlebeckineq}
d_{S}(u) > \underset{w \neq u}{\sum} x_{S}(w) +\Delta - \frac{d(u)}{2}\ .
\end{equation}
Note that since the cut $S$ is stable, $x_S(w)\ge 0$ holds for all vertices, and if $x_S(w)>0$, then $x_S(w)\ge 1/2$. Let us study the value of $X=\underset{w \neq u}{\sum} x_{S}(w)$. 

\textbf{Case 1.} If $X\ge x_S(u')+1$, then we have, via \eqref{eq:localbottlebeckineq}, that 
\begin{align*}
d_S(u)&> x_S(u') + 1 + \Delta-d(u)/2\\
&\ge x_S(u') + 1 + d(u')/2=d_S(u')+1\ge d_{S\oplus v}(u')\ ,
\end{align*}
where we use $\Delta-d(u)/2\ge \Delta/2\ge d(u')/2$, and $d_S(u')=x_S(u')+d(u')/2$. This, however, implies that $d_{S\oplus v}(u)= d_S(u)-1\ge d_{S\oplus v}(u')$, which contradicts that $u'$ is the unique critical vertex of $S\oplus v$; therefore, $X\le x_S(u')+1/2$.

\textbf{Case 2.} If $X=x_S(u')+1/2$, then there is a vertex $u''\notin \{u,u'\}$ such that $x_S(u'')=1/2$, and for every $w\notin \{u,u',u''\}$, $x_S(w)=0$.  
Using this in \eqref{eq:localbottlebeckineq}, we have 
\[
d_S(u)>x_S(u')+1/2+\Delta-d(u)/2\ .
\]
Let us show that $\Delta-d(u)/2\ge d(u')/2+1/2$, which would imply that the right-hand side above is at least $d_S(u')+1$. Assume, towards a contradiction, that $\Delta< (d(u')+d(u) + 1)/2$, which implies that $d(u)=d(u')=\Delta$. Recall that $x_S(u'')=1/2$, implying that $d(u'')$ is odd. For every vertex $w\notin\{u,u',u''\}$, we have $x_S(w)=0$, hence $d(w)$ is even. This implies that we have an odd number of odd-degree vertices (since $d(u)=d(u')$), which is impossible. Thus, we again have $d_S(u)>d_S(u')+1\ge d_{S\oplus v}(u')$, which contradicts the fact that $u'$ is the unique critical vertex of $S\oplus v$.

\textbf{Case 3.} We conclude that $X=x_S(u')$. It immediately follows that all vertices $w\notin \{u,u'\}$ have $x_S(w)=0$. We also have the following property.
\begin{claim}\label{claim:long1}
Every crossing edge is adjacent to either $u$ or $u'$, the latter only when $x_S(u')>0$.
\end{claim}
\begin{proof}
Let $w\notin S\cup \{u,u'\}$. 
We have that $x_S(w)=0$ and $w\notin S$, so by our assumption, $S\oplus w$ is a stable cut. The latter implies that there cannot be an edge $\{w,w'\}\in E$ for $w\notin S$, $w'\in S$, such that $x_S(w')=0$, since we would then have $d_{S\oplus w}(w')<d(w')/2$, contradicting the assumption that $S\oplus w$ is stable. However, each vertex $w$ with $x_S(w)=0$ has exactly half of its neighbors in the opposite side of the cut. This can only happen if $w$ is isolated or its only neighbors in the opposite side are $u$ and, possibly, $u'$. The latter can only happen when $x_S(u')>0$.
\end{proof}
To proceed, we consider two cases for $x_S(u')$.

\textbf{Case 3.1.} Consider the case $x_S(u')=0$; then, the only vertex $w$ with $x_S(w)>0$ is $u$. The argument above implies that there is no vertex $w$ on the same side of the cut as $u$, since otherwise $w$ would be an isolated vertex. Thus, we have $S=\{u\}$, and each vertex in $V- S$ is adjacent to $u$ and has exactly one neighbor in $V- S$, i.e., for every $w\neq u$, $d(w) = 2$. Since $\Delta> 2$, we have that $d_S(u) = d(u) > 2$, therefore $d_{S\oplus v}(u) \ge 2$ and it follows that $d_{S\oplus v}(u)=\max_{w\in V}d_{S\oplus v}(w)$. This, however, contradicts  the assumption that $u'$ is the unique critical vertex in $S\oplus v$, and that $u'\neq u$.

\textbf{Case 3.2.} If $x_S(u')>0$, then there are two possibilities: either $u'$ is on the same side of the cut as $u$ or it is on the other side. 

\textbf{Case 3.2.1.} If $u,u'$ are on the same side of the cut, then by a similar reasoning as above, $S=\{u,u'\}$. In $S\oplus v$, we have $d_{S\oplus v}(u)=d_S(u)-1$, and $d_{S\oplus v}(u')\le d_{S}(u')\le d_S(u)-1$, where the last inequality holds since $u$ is the unique critical vertex of $S$. We conclude that $d_{S\oplus v}(u')\le d_{S\oplus v}(u)$, which contradicts the assumption that $u'$ is the unique critical vertex of $S\oplus v$, and $u'\neq u$.

\textbf{Case 3.2.2.} It remains to consider the case when $x_S(u')>0$, and $u$ and $u'$ are on opposite sides of the cut $S$. Recall that in this case, $u,u'$ are the only vertices with positive excess, hence, by the choice of  $v$, if $Q=(N(u)\setminus N(u'))\setminus S\neq \emptyset$, then $v\in Q$. We show that $Q=\emptyset$. Assume the opposite, i.e., $Q\neq \emptyset$, and $v\in Q$. Since $d(u)$ and $d(u')$ have the same parity (all other degrees in the graph are even), we have that $x_S(u)\ge x_S(u')+1$. Since $\{v,u\}\in E$ and $\{v,u'\}\notin E$, we have $x_{\sv}(u')=x_{S}(u')$, and $x_{\sv}(u)=x_{S}(u)-1$, implying that $x_{\sv}(u)\ge x_{\sv}(u')$, which contradicts to the assumption that $u\neq u'$ and $u'$ is the unique critical vertex in $\sv$. Thus, we have $Q=\emptyset$: every neighbor of $u$ in $V\setminus S$ is also a neighbor of $u'$. 
Since for every node $w\notin\{u,u'\}$, $x_S(w)=0$, it follows from Claim~\ref{claim:long1}, that $w$ has a single crossing edge: $\set{w,u}$, if $w\notin S$ or $\set{w,u'}$, if $w\in S$. Thus, every vertex in $S$ is adjacent to $u'$, and every vertex in $V\setminus S$ is adjacent to $u$. Since $Q=\emptyset$, we see that every vertex in $V\setminus S$ is also adjacent to $u'$ (and to no other vertex in $V\setminus S$); hence $\Delta=d(u')\ge d(u)$. 

\begin{claim}
$d(u)=d(u')=\Delta$, and for every $w\notin\{u,u'\}$, $N(w)=\{u,u'\}$.
\end{claim}
\begin{proof}
We only need to show that $d(u)=\Delta$: then, both $u$ and $u'$ are adjacent to all other vertices, and By Claim~\ref{claim:long1}, every such vertex has degree 2. Assume, towards a contradiction, that $d(u) < \Delta$. By \eqref{eq:localbottlebeckineq}, we have that $x_S(u) > x_S(u') + \Delta - d(u)$, hence $x_S(u) > x_S(u') + 1$. Since $d(u),d(u')$ have the same parity (all other degrees are even)  $x_S(u) - x_S(u')$ is an integer, and hence  $x_S(u) - x_S(u') \ge 2$. It holds that $x_{S\oplus v}(u) \ge x_S(u) - 1$, and $x_{S\oplus v}(u') \le x_S(u') + 1$, i.e., $x_{S\oplus v}(u) \ge x_{S\oplus v}(u')$, which contradicts the assumption that $u'$ is the unique critical vertex of $S\oplus v$, and that $u'\neq u$.
\end{proof}

Recall that for every $w\notin \set{u,u'}$, exactly one of $\set{w,u}$ and $\set{w,u'}$ crosses $S$. Since $u$ is the single critical vertex in $S$, we have $d_S(u)>d_S(u')$. If $u$ and $u'$ are not adjacent, then $d_S(u)+d_S(u')=\Delta$. The latter implies that $d_S(u')<\Delta/2=d(u')/2$, contradicting to the assumption that $S$ is a stable cut. Thus, $\{u,u'\}\in E$, and $d_S(u)+d_S(u')=\Delta+1$, since the edge $\set{u,u'}$ is counted in both $d_S(u)$ and $d_S(u')$. By the reasoning above, $d_S(u')\ge \Delta/2$, and hence $d_S(u)\le \Delta/2+1$. Hence, $\Delta$ is even, and $d_S(u')=\Delta/2$, implying that $x_S(u')=0$, which is a contradiction.

Since we got a contradiction in all cases, we conclude that $u=u'$, as claimed at the beginning. This completes the proof.
\end{proof}

We are now ready to prove the main statement about our algorithm.

\begin{proof}[Proof of Theorem~\ref{thm:Greedy_ftc}] By Lemma~\ref{lem:specific_case_half}, \textsc{StabilizeCut} gives a $1/2$-approximation when $\Delta\le 2$, so we focus on $\Delta\ge 3$. First, let us show correctness of the algorithm, that is, if neither of type-$0$, type-$1$ or build-up steps applies, then a type-$2$ step can be applied. Let $\sol_{i}$ denote the cut in the beginning of  iteration $i$. Let $i$ be an iteration where none of the first three steps applies. Assume w.l.o.g. that the critical vertex of $\sol_{i}$ is in $\sol_{i}$. Then, $\sol_i$ is stable, for every $w\notin \sol_{i}$ with $x_{\sol_i}(w) = 0$, $\varphi(\sol\oplus w) < (m-\Delta)/2$, and $\sol\oplus w$ is stable (the latter holds by Lemma~\ref{lem:opluscut}, as no build-up step applies). 
Thus, Lemma~\ref{lem:has_greedy_step} holds, so there is a type-$2$ step.

We show that the algorithm terminates within $4m+2$ iterations, giving a cut $\sol$ with 
$\varphi(\sol) \geq (m-\Delta)/2$, which by Observation~\ref{obs:greedy:max_ft_size} is a $1/2$-approximation. 
Assume, towards a contradiction, that the algorithm does not terminate within $4m+2$ iterations, i.e., no  type-$0$ is applied. 
 Since we always have $\sol_{i+1} = \sol_i\oplus v$ with $d_{\sol_i}(v) \leq d(v)/2$, Observation~\ref{eq:oplusdiff} implies that $C_{\sol_i} \le C_{\sol_{i+1}}$. We also have $\varphi(\sol_i)\leq \varphi(\sol_{i+1})$, which in type-$2$ and build-up steps holds by definition, and in type-$1$ steps holds by Lemma~\ref{lem:k_greedy_step}.

Next, we show that in every consecutive pair of iterations, either the cut size or the \aftv  strictly increases. 
Formally, for every $i$, either $C_{\sol_{i}} < C_{\sol_{i+2}}$ or $\varphi(\sol_{i}) < \varphi(\sol_{i+2})$ holds. If either one of the two iterations is a type-$1$ step then $C_{\sol_{i}} < C_{\sol_{i+2}}$, since the cut size never decreases, while it increases in a type-$1$ step, by Lemma~\ref{lem:k_greedy_step}. If neither of the two steps is a type-$1$ step, then the first one, $i$, must be a type-$2$ step: otherwise, it would be a build-up step, which has to be followed by a type-$1$ step. By the definition of a type-$2$ step,  $\varphi(\sol_{i}) < \varphi(\sol_{i+1})\le \varphi(\sol_{i+2})$.

It follows that after $4m+2$ iterations, $\max\{C_{\sol},\varphi(\sol)\} > m$, which gives a contradiction.
\end{proof}

\section{Fault Tolerance Against an Oblivious Adversary}\label{sec:Oblivious}

We give an algorithm that approximates the fault tolerant \mc against the oblivious adversary with (constant) $k$ faults within an $\agw$-approximation factor. The main idea is to frame the problem as a linear program (LP) with  an exponential number of variables, then reduce the number of variables using a solution of its dual (with an exponential number of constraints but a polynomial number of variables). The dual is approximately solved by the ellipsoid algorithm together with an approximate separation oracle that is given by  a \mc algorithm. A similar approach has been used,  e.g. in \cite{DBLP:conf/soda/JainMS03}, for an unrelated problem. 
\thmoblagw*

For simplicity, we present the algorithm for a single fault, and then show how to extend it to any constant number $k$ of faults.
The \oblftc problem can be formulated as the following LP, \eqref{primal_lp}, with an exponential number of variables.

\newcommand{\mathbox}[3][\mathop]{
  #1{\eqmakebox[#2]{$\displaystyle#3$}}
}

\begin{align}
\max\quad  &\mathbox{B}{\sum\limits_{S\subseteq V}} P_S\cdot \sum\limits_{ e\in \delta(S)}w_e - Z&& \label{primal_lp}\tag{$Primal_1$}\\
\text{s.t. } \quad & \mathbox{B}{\sum\limits_{\mathclap{S\subseteq V}}} P_S\cdot {\sum\limits_{v: \set{u,v}\in\delta(S)}}w_{\set{u,v}} \leq Z  \qquad &&\forall u \in V \label{primal_cut_c}\\
    &\mathbox{B}{\sum\limits_{S\subseteq V}} P_S \leq 1&& \label{primal_prob_c}\\
    &~0 \leq P_S \qquad&&\forall S\subseteq V \label{primal_basic_c}
\end{align}

The variable $P_S$ represents the probability assigned to the cut $S\subseteq V$. 
The variable $Z$ represents the expected weight that the adversary removes from the graph. Constraints (\ref{primal_prob_c}-\ref{primal_basic_c}) make $P_S$ a probability distribution. In (\ref{primal_cut_c}), for each vertex $u$, we bound by $Z$ the expected weight that is removed from the cut when $u$ fails. 
To see that the left hand side is indeed the expected removed weight, note that  it equals $\sum_{S\subseteq V} P_S\cdot d_S(u)$.

Consider the dual problem of the LP above, \eqref{dual_lp}: 

\begin{align}
\min ~~~ & Y \label{dual_lp}\tag{$Dual_1$} & \\
\text{s.t.} ~~~~ & \mathbox{D}{\displaystyle\sum\limits_{\substack{\set{u,v}\in \delta(S)}}}w_{\set{u,v}} - \displaystyle\sum\limits_{u\in V} X_u \displaystyle\sum\limits_{\substack{v:\set{u,v}\in \delta(S)}}w_{\set{u,v}} \leq Y  &  \forall S\subseteq V\label{dual_cut_c}\\
&\mathbox{D}{\sum_{u \in V}} X_u \leq 1 \label{dual_sum_c} & \\
 &~~~~ 0 \leq X_u & \forall u \in V\label{dual_basic_c}
\end{align}

The dual LP captures the following problem: The adversary picks a distribution over the vertices, and the algorithm picks a cut (depending on the choice of the adversary). 
The goal of the adversary is to choose its distribution (without knowing the cut choice of the algorithm) so as to minimize the expected cut size after a random failure from its distribution. 

The dual LP \eqref{dual_lp} has  an exponential number of constraints but only $|V|+1$ variables. Such LPs can be solved efficiently via the \emph{ellipsoid method} \cite{ellipsoid}, 
given an efficient \emph{separation oracle}. The latter is 
an algorithm that given an  assignment of values to the variables of the LP, reports a violated constraint if the assignment is infeasible, or otherwise reports that it is feasible. For the particular case of \eqref{dual_lp}, the ellipsoid algorithm can be viewed as a binary search over the values of $Y$,  
such that in each stage (fixed $Y$), a black-box procedure does a polynomial number of queries to a given separation oracle, and either reports the first solution $\{X_u\}_{u\in V}$ it finds such that $\{X_u\}_{u\in V},Y$ is feasible according to the oracle, or reports that there is no such solution.

Let us see what a separation oracle looks like in our case. For given values $\{X_u\}_{u\in V},Y$, let $G'=(V',E',w')$ be the graph with weights $w'_{\{u,v\}}=(1-X_u-X_v)w_{\{u,v\}}$. With this notation,  constraint \eqref{dual_cut_c} becomes $C_{S,G'}\le Y$. In order to see if  a given assignment of variables is feasible, it thus suffices to find a maximum weight cut $S^*$ in $G'$ and test if $C_{S^*,G'}\le Y$. Since  \mc is hard to solve exactly, we use an \emph{approximate separation oracle}. Given $\{X_u\}_{u\in V},Y$, it immediately returns the constraint \eqref{dual_sum_c}, if it is violated, and otherwise computes  a cut $S_{ALG}$ in $G'$ using a derandomized variant of the Goemans-Williamson algorithm~\cite{DBLP:journals/jacm/GoemansW95, DBLP:journals/siamcomp/MahajanR99}, which we denote by \textsc{Derandomized-Goemans-Williamson}. If the size of the cut is larger than $Y$, it returns the violated constraint \eqref{dual_cut_c} corresponding to $S_{ALG}$, otherwise it reports that the solution is feasible. In Algorithm~\ref{alg:oracle} we give the pseudocode of our approximate separation oracle.

\begin{algorithm}
\DontPrintSemicolon
\caption{Approximate separation oracle}
\label{alg:oracle}

\textbf{Input:} $\set{X_u}_{u\in V}, Y,G$\; 
\If{$\sum\limits_{u \in V} X_u > 1$}{
    \Return violated constraint  $\sum\limits_{u \in V} X_u \leq 1$
    }

$S_{ALG}\gets$ \textsc{Derandomized-Goemans-Williamson}($G'$)\;

\eIf{$C_{S_{ALG},G'} > Y$}{

\Return violated constraint for subset $S_{ALG}$
}
{\Return \texttt{feasible}
}
\end{algorithm}

The following lemma shows that the oracle always answers correctly if the assignment is feasible, and even if it incorrectly outputs \texttt{feasible}, the assignment is nearly feasible.  

\begin{lemma}\label{lem:separation_oracle}
Given an assignment $\set{X_u}_{u\in V}, Y$ to the variables  in \eqref{dual_lp} as input to the separation oracle in Algorithm~\ref{alg:oracle}, it holds that:
\begin{enumerate}
    \item if the assignment is feasible, then the oracle returns \texttt{feasible}, 
    \item if the assignment is infeasible, then either the oracle outputs a violated constraint, or 
    reports \texttt{feasible}, in which case 
    $\set{X_u}_{u\in V}, Y/\agw$ is feasible.
\end{enumerate}
\end{lemma}

\begin{proof}
Let $\set{X_u}_{u\in V}, Y$ be an assignment to the variables of \eqref{dual_lp}. If it is feasible, then it holds that $\sum\limits_{u \in V} X_u \leq 1$, and in addition every $S\subseteq V$ satisfies $C_{S,G'}\le Y$, therefore the oracle returns \texttt{feasible}.

If the assignment is infeasible, there are two cases. If $\sum\limits_{u \in V} X_u > 1$, the oracle returns this violated constraint. Otherwise, there is a subset $S'\subseteq V$ such that $C_{S',G'}>Y$. Let $S^*$ be an optimal solution for \mc on $G'$, and note that $C_{S^*, G'} > Y$. If $\agw\cdot C_{S^*,G'} > Y$, then we also have that $C_{S_{ALG}, G'} > Y$ (since $S_{ALG}$ is an $\agw$-approximate \mc), 
and the oracle returns the violated constraint for $S_{ALG}$.

Otherwise, $C_{S^*, G'} \leq Y/\agw$. Since $S^*$ is an optimal solution for \mc on $G'$, it follows that for every $S\subseteq V$, it holds that $C_{S,G'}\le Y/\agw$, i.e., the solution $\set{X_u}_{u\in V}, Y/\agw$ is feasible.
\end{proof}

It is not hard to see that the application of the ellipsoid algorithm on \eqref{dual_lp} takes a polynomial time (i.e., at most as much time as it would take with an exact separation oracle), since our approximate oracle is (possibly) incorrect only on \emph{the last} call from the ellipsoid algorithm (for a given $Y$), when it incorrectly reports a solution as feasible.

The output of the ellipsoid algorithm/binary search is an assignment $\{X_u\}_{u\in V},Y$ to the variables of \eqref{dual_lp} such that $\{X_u\}_{u\in V},Y$ is feasible according to the oracle, while $Y-\epsilon$ is infeasible with every assignment to the $X$ variables, where $\epsilon$ is the precision of the binary search. As observed above, we have that  $\{X_u\}_{u\in V},Y/\agw$ is feasible, and 
it follows that if $Y^*$ is the optimal value of \eqref{dual_lp}, then  $Y-\epsilon\le Y^*\le Y/\agw$. Since the ellipsoid algorithm queries the oracle a polynomial number of times, there is a set $\mathcal{H}\subseteq 2^V$ of a polynomial number of cuts $S$, for which  constraint \eqref{dual_cut_c} is queried. Consider a modified variant 
of \eqref{dual_lp}, called \eqref{ndual_lp}, where only constraints of cuts in $\mathcal{H}$ are present:

\begin{align}
\min~~~ & Y \label{ndual_lp}\tag{$Dual_2$} & \\
\text{s.t.} ~~~~ & \mathbox{E}{\sum\limits_{\substack{\set{u,v}\in \delta(S)}}}w_{\set{u,v}} - \sum\limits_{u\in V} X_u \sum\limits_{\substack{v:\set{u,v}\in \delta(S)}}w_{\set{u,v}} \leq Y  &  \forall S\in\mathcal{H}\nonumber\\
&\mathbox{E}{\sum_{u \in V}} X_u \leq 1 \nonumber & \\
 &~~~~ 0 \leq X_u & \forall u \in V\nonumber
\end{align}

Let $Y^*_2$ be the optimal value of \eqref{ndual_lp}. Note that $Y^*_2\le Y^*$. Note also that the ellipsoid algorithm returns exactly the same solution $\{X_u\}_{u\in V},Y$, when executed on \eqref{dual_lp} and \eqref{ndual_lp} (since our algorithm is deterministic, and only constraints in $\mathcal{H}$ are queried); hence, we have $Y-\epsilon\le Y^*_2$.
Finally, let us consider the primal LP corresponding to \eqref{ndual_lp}:

\begin{align}
\max\quad  &\mathbox{B}{\sum\limits_{S\in \mathcal{H}}} P_S\cdot \sum\limits_{ e\in \delta(S)}w_e - Z&& \label{nprimal_lp}\tag{$Primal_2$}\\
\text{s.t. } \quad & \mathbox{B}{\sum\limits_{S\in\mathcal{H}}} P_S\cdot {\sum\limits_{v: \set{u,v}\in\delta(S)}}w_{\set{u,v}} \leq Z  \qquad &&\forall u \in V \nonumber\\
    &\mathbox{B}{\sum\limits_{S\in \mathcal{H}}} P_S \leq 1&& \nonumber\\
    &~0 \leq P_S \qquad&&\forall S\in \mathcal{H} \nonumber
\end{align}
Note that \eqref{nprimal_lp} is obtained from \eqref{primal_lp} by removing variables $P_S$ with $S\notin \mathcal{H}$ (i.e., setting $P_S=0$).

The new primal has polynomially many constraints and variables, so can be solved in polynomial time. From the arguments above, we have that its optimal value $Y^*_2$ satisfies $Y-\epsilon\le Y^*_2\le Y^*\le Y/\agw$. Recalling that $Y^*$ is the optimal value for the original LP, we see that $Y^*_2$ is a $\agw$-approximation (with any polynomial precision $\epsilon$).

\paragraph*{Extending the proof of Theorem~\ref{thm:obl_agw} for \texorpdfstring{$k$}{Lg} failures}
 In order to extend the algorithm to $k$ failures, for a constant $k\in \mathbb{N}$, all we need to do is to slightly generalize the primal and dual LPs, while the overall structure stays the same. The extension of \eqref{primal_lp} to the case of $k$ faults is as follows:

\begin{align} 
\max\quad  &\mathbox{B}{\sum\limits_{S\subseteq V}} P_S\cdot \sum\limits_{ e\in \delta(S)}w_e - Z&& \nonumber\\
\text{s.t. } \quad & \mathbox{B}{\sum\limits_{\mathclap{S\subseteq V}}} P_S\cdot {\sum\limits_{e\in\delta(S), e\cap F\neq\emptyset}}w_{e} \leq Z  \qquad &&\forall F\in\binom{V}{k} \nonumber\\
    &\mathbox{B}{\sum\limits_{S\subseteq V}} P_S \leq 1&& \nonumber\\
    &~0 \leq P_S  \qquad&&\forall S\subseteq V \nonumber
\end{align}

The corresponding dual problem is as follows.
\begin{align}
\min ~~~ & Y \nonumber & \\
\text{s.t.} ~~~~ & \mathbox{H}{\sum\limits_{\substack{e\in \delta(S)}}}w_{e} - \sum\limits_{F\in\binom{V}{k}} X_F \sum\limits_{\substack{e \in \delta(S), e\cap F\neq\emptyset}}w_{e} \leq Y  &  \forall S\subseteq V\nonumber\\
&\mathbox{H}{\sum_{F\in\binom{V}{k}}} X_F\leq 1 \nonumber & \\
 &0 \leq X_F & \forall F \in \binom{V}{k}\nonumber
\end{align}

The separation oracle is similar to Algorithm~\ref{alg:oracle}, but defines the weight function of $G'$ as $w'_{\set{u,v}} = \bigg(1- \displaystyle\sum\limits_{\set{u,v}\cap F \neq \emptyset}X_F\bigg)w_{\set{u,v}}$. The rest of the algorithm is the same, and the proof is similar.

\FloatBarrier

\section{Hardness of Approximation}\label{sec:hardness}
In this section we show that assuming the Unique Games Conjecture, one cannot approximate \ftc and \oblftc within a factor greater than $\agw$. Formally, we prove the following:

\hardnessobl*

In both cases, given an unweighted instance $G$ of \mc, we construct an unweighted graph $G'$, according to Algorithm~\ref{alg:mc_by_ft}: we take the disjoint union of $G$ with a star with $n=|V|$ leaves and a center $u^*$, and add an edge joining $u^*$ to an arbitrary vertex $v_1\in V$. 
This completes the construction of $G'$ (see Figure~\ref{figure:construct_G_prime}). Clearly, this is a polynomial construction. 

\begin{figure}
\centering
\includegraphics[width=8.5cm, scale=0.25]{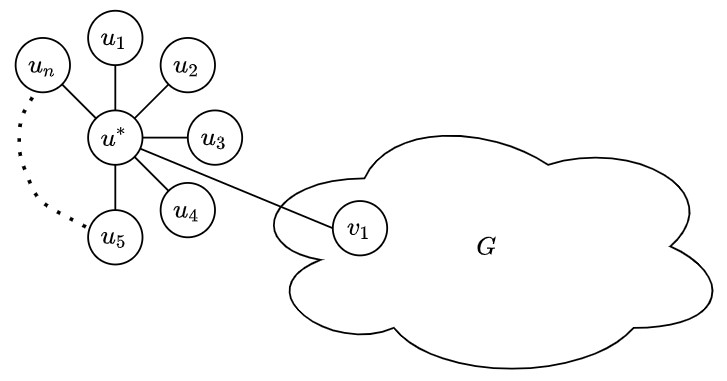}
\caption{The Construction of $G'$}
\label{figure:construct_G_prime}
\end{figure}

\begin{algorithm}
\DontPrintSemicolon
\caption{Approximate \mc Using \ftc}
\label{alg:mc_by_ft}

\textbf{Input:} $G=(V,E)$, $V=\{v_1,\dots,v_n\}$\;

Let $G' = (V',E')$ for $V' = V\cup \set{u^*, u_1, \dots, u_n}$, $E' = E \cup \set{\set{u^*, u_i}:\text{ }i\in [n]} \cup \set{\set{u^*, v_1}}$\;

\Return $G'$\;
\end{algorithm}

Below, we show for each kind of adversary how to translate a given (approximate) solution to \ftc or \oblftc in $G'$ into a solution to \mc in $G$, which would imply the corresponding inapproximability results, using the fact that \mc is hard to approximate within a factor better than $\agw$ \cite{DBLP:journals/siamcomp/KhotKMO07}. We use the following simple observation.\label{tryE1}

\begin{observation}\label{obs:hardnesssimple}
Let $S\subseteq V$ be a cut in $G$, and $S'=S\cup\{u^*\}$. It holds that in $G'$, $u^*$ is a critical vertex of $S'$, i.e., $\varphi(S')=C_{S'-u^*, G'-u^*}$. For every cut $S''\subseteq V'$, we have $C_{S''-u^*,G'}=C_{S''\cap V,G}$.
\end{observation}
The proof follows from the fact that for every vertex $v\in V'$, $d_{S'}(v)\le n\le d_{S'}(u^*)$, and that all edges in $G'-u^*$ belong to $G$.

\subsection{Adaptive Adversary}

First, let us observe that the optimal values of \mc in $G$ and \ftc in $G'$ are equal.\label{tryE2}

\begin{lemma}\label{lem:optft_optmc}
Let $S^*_{ft}$ be an optimal  \ftc in $G'$ and $S^*_{mc}$ be an optimal \mc in $G$. It holds that $\varphi(S^*_{ft}, G') = C_{S^*_{mc}, G}$.

\end{lemma}
\begin{proof}

First, we show that $\varphi(S^*_{ft}, G') \leq C_{S^*_{mc}, G}$. This follows since  $C_{S^*_{ft}-u^*,G'\sm u^*}\ge \varphi(S^*_{ft}, G')$, and  by Observation~\ref{obs:hardnesssimple},  $C_{S^*_{ft}-u^*,G'\sm u^*}=C_{S_{ft}\cap V,G}\le C_{S^*_{mc}, G}$.

Next, let us show that  $\varphi(S^*_{ft}, G') \geq C_{S^*_{mc}, G}$. Let $S = S^*_{mc} \cup \set {u^*}$. By Observation~\ref{obs:hardnesssimple}, we have that $\varphi(S^*_{ft}, G')\ge\varphi(S, G') = C_{S\sm u^*, G'\sm u^*}=C_{S\cap V,G}=C_{S^*_{mc},G}$. 
\end{proof}

\begin{proof}[Proof of Theorem~\ref{thm:hardness_obv} for an Adaptive Adversary] Assume, for contradiction, that we have an $\alpha$-approximation algorithm for \ftc, for $\alpha>\agw$. Let us construct an $\alpha$-approximation algorithm for \mc. For any input graph $G$, construct the graph $G'$, as per Algorithm~\ref{alg:mc_by_ft}. As assumed, we can compute an $\alpha$-approximate \ftc  $S_{ft}$ in $G'$. Let $S^*_{ft}$ and $S^*_{mc}$ be   optimal solutions for \ftc in $G'$ and \mc in $G$ (resp.).
By Observation~\ref{obs:hardnesssimple} and Lemma~\ref{lem:optft_optmc}, we have  $C_{S_{ft}\cap V,G}=C_{S_{ft}-u^*, G'-u^*}\ge \varphi(S_{ft}, G')\ge \alpha\cdot \varphi(S^*_{ft}, G')= \alpha\cdot C_{S^*_{mc}}$. 
Thus, we get an $\alpha$-approximation algorithm for \mc, which is impossible under the Unique Games Conjecture and $NP\neq P$. The proof extends to randomized algorithms as well (here we also assume $NP\neq BPP$, which together with the Unique Games Conjecture excludes better than $\agw$-approximation algorithms, including randomized, for \mc). 
\end{proof}

\subsection{Oblivious Adversary}

Here, we rely on the fact that under the assumption of the Unique Games Conjecture and $NP\neq BPP$, there is no randomized algorithm that outputs a better than $\agw$-approximation for \mc with constant probability.

Again, we begin by showing that the optimal values for \oblftc in $G'$ and \mc in $G$ are equal.

\begin{lemma}\label{lem:mc_by_oblftc}
Let $\mathcal{D}^*$ be the distribution of an optimal \oblftc in $G'$, 
and $S^*_{mc}$ be an optimal \mc in $G$. It holds that $\mu(\mathcal{D}^*, G')=C_{S^*_{mc}, G}$. 
\end{lemma}
\begin{proof}
Let $\sol = S^*_{mc} \cup \set {u^*}$, and let $\mathcal{D}$ be the distribution that assigns   probability $1$ to $\sol$ and probability 0 to all other cuts. By Observation~\ref{obs:hardnesssimple}, $u^*$ is a critical vertex, hence for every vertex $v\in G'$, we have 
$\mexpl{{S \sim \mathcal{D}}}{C_{S\sm u^*, G'\sm u^*}} = C_{\sol\sm u^*, G'\sm u^*}\le C_{\sol\sm v, G'\sm v}=\mexpl{{S \sim \mathcal{D}}}{C_{S\sm v, G'\sm v}}$. 
Using Observation~\ref{obs:hardnesssimple} again, we have $\mu(\mathcal{D}^*, G')\ge \mu(\mathcal{D}, G')\ge C_{\sol-u^*,G'-u^*}=C_{S^*_{mc},G}$. 

Next, by Observation~\ref{obs:hardnesssimple}, we have $C_{S-u^*,G'-u^*}=C_{S\cap V,G}\le C_{S^*_{mc}, G}$, for every cut $S$ from the support of $\mathcal{D}^*$, which implies that  $\mu(\mathcal{D}^*, G')\le \mexpl{{S \sim \mathcal{D}^*}}{C_{S\sm u^*, G'\sm u^*}}\le C_{S^*_{mc}, G}$. This  completes the proof.
\end{proof}

\begin{proof}[Proof of Theorem~\ref{thm:hardness_obv} for an Oblivious Adversary]
Assume, for a contradiction, that we have an $\alpha$-approximation algorithm for \oblftc, for $\alpha=\agw+\epsilon>\agw$. We design a randomized approximation algorithm for \mc. Let $G$ be an input to \mc. Construct the graph $G'$ as per Algorithm~\ref{alg:mc_by_ft}. 
Let $\mathcal{D}$ be the distribution of an $\alpha$-approximate \oblftc in $G'$.
By Lemma~\ref{lem:mc_by_oblftc}, we have $\mexpl{{S \sim \mathcal{D}}}{C_{S\sm u^*, G'\sm u^*}}\ge \mu(\mathcal{D},G')\ge  \alpha \cdot C_{S^*_{mc}, G}$, where $S^*_{mc}$ is a \mc in $G$. By Observation~\ref{obs:hardnesssimple}, it holds that $C_{S_{ft}-u^*, G'-u^*}=C_{S_{ft}\cap V, G}\le C_{S^*_{mc}}$ for every cut $S_{ft}$ in the support of $\mathcal{D}$. 
Letting $p=\mpr{C_{S\sm u^*, G'\sm u^*}\ge (\alpha-\epsilon/2) C_{S^*_{mc}}}$, we have   \[
\alpha\cdot C_{S^*_{mc},G}\le \mexpl{{S \sim \mathcal{D}}}{C_{S\sm u^*, G'\sm u^*}}\le p\cdot C_{S^*_{mc}, G} + (1-p)\cdot (\alpha-\epsilon/2) C_{S^*_{mc},G}\ , 
\]
implying that $p\ge \epsilon/2$. Thus, for a random cut $S_{ft}$ sampled  from $\mathcal{D}$, it holds that $S_{ft}\cap V$ is an $(\alpha-\epsilon/2)$-approximation to \mc, with probability $\epsilon/2$, where $\alpha-\epsilon/2>\agw$. This contradicts to our assumption about the Unique Games Conjecture and $NP\neq BPP$.
\end{proof}

\section{The Approximation Factor of a Random Cut}\label{sec:randomcut} In this section, we study the approximation provided by a random cut for the fault tolerant \mc problem, where a random cut is obtained by including each vertex in the cut independently, with probability $1/2$.

In the case of an adaptive adversary, we show that a random cut cannot achieve an $\alpha$-approximation for weighted graphs, with $\alpha > 1/4$. Nonetheless, we show that if the input is a \emph{connected} unweighted graph with sufficiently many vertices, then a random cut gives a $(1/2-\epsilon)$-approximation. For an oblivious adversary, we show that the uniform distribution over all cuts, which can be seen as a randomized algorithm that outputs a random cut as described above, gives a $1/2$-approximation for weighted instances and many faults.

\subsection{A Negative Result for Adaptive Adversary and Weighted Instances
}
Here we show that a random cut cannot achieve an approximation factor that is greater than $1/4$ for weighted \ftc.\label{tryF1}

\begin{theorem}\label{thm:random_cut_weight_ftc}
For every $\epsilon > 0$ and $n\ge 3/4\epsilon$, there is an $n$-vertex graph $G$ such that $\mexp{\varphi(S)}\leq(1/4+\epsilon)\varphi(S^*)$, where $S$ is a uniformly random cut $S$,  and $S^*$ is an optimal \ftc in $G$.
\end{theorem}

\begin{proof}
Let $n \geq 4$, and consider the weighted graph $G=(V,E,w)$ with $V = \set{v_0, v_1,\dots, v_{n-1}}$, $E = \{e_i=\{v_{i}, v_{i+1 \mod{n}}\} : i=0,1,\dots,n-1\}$, and the following weight function: $w_{e_0} = w_{e_2}=n(n-3)$ (we call these edges \emph{heavy}), and $w_{e_i} = 1$, for $i\notin \{0,2\}$ (we call these edges \emph{light}). Thus, $G$ is a weighted cycle with edges of weight $1$, except for two non-adjacent edges of weight $n(n-3)$ (See Figure~\ref{figure:construct_G_weighted_random_cut}).
\begin{figure}
\centering
\includegraphics[width=6cm, scale=0.25]{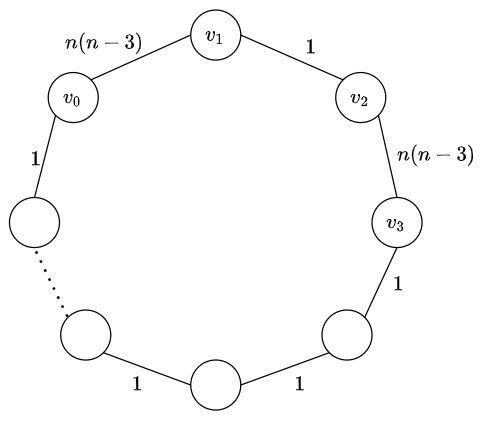}
\caption{The Construction for Theorem~\ref{thm:random_cut_weight_ftc}.}
\label{figure:construct_G_weighted_random_cut}
\end{figure}

Let $S^*$ be an optimal  \ftc in $G$. Note that $\varphi(S^*) \ge n(n-3)$, since, e.g., $\varphi(\{v_1,v_2\})=n(n-3)$. Let $S$ be a random cut, and let $\mathcal{E}$ be the event that both edges $e_0$ and $e_2$ cross $S$. We have $\mpr{\mathcal{E}}=1/4$, since the probability for each edge to cross is $1/2$, and $e_0,e_2$ are  disjoint, hence independent. 
 Given $\neg\mathcal{E}$, i.e., at least one of $e_0,e_2$, say, $e_0$,  does not cross $S$, we have that $\varphi(S)\le n-3$, since, when we fail an endpoint of $e_2$, say, $v_2$, there are no heavy crossing edges left. Thus, we have
 \begin{align*}
 \mexp{\varphi(S)} &= \mexp{\varphi(S) \mid \mathcal{E}}\cdot \mpr{\mathcal{E}} + \mexp{\varphi(S) \mid \neg\mathcal{E}}\cdot \mpr{\neg\mathcal{E}}\\ 
 &\leq (1/4)\varphi(S^*) + (3/4)(n-3)\\ 
 &= \varphi(S^*)\cdot (1/4+3(n-3)/4\varphi(S^*))\\
 &\le \varphi(S^*)\cdot (1/4+3/4n)\ .
 \end{align*}
We get the desired approximation, since $n\ge 3/4\epsilon$.
\end{proof}

\subsection{Adaptive Adversary and Unweighted Instances}
Here we show that for every $0 < \epsilon < 1/16$, a random cut can get a $(1/2-\epsilon)$-approximation for unweighted \ftc for every graph $G$ with large enough $n$.

\begin{theorem}\label{thm:random_cut_unweight_ftc}
Let $0 < \epsilon < 1/16$, and let $S$ be a random cut in an unweighted graph $G=(V,E)$. If $G$ is connected and $n=|V|$ is large enough (w.r.t. $\epsilon$) then $S$ is a $(1/2-\epsilon)$-approximation for \ftc with a single fault in $G$.
\end{theorem}

\begin{proof}
\def \eps {\epsilon}
Fix a constant $0 < \epsilon < 1/16$. We assume that $n$ is sufficiently large w.r.t. $\eps$. Let $c=48/\eps^4$. We also assume that $G$ is connected. In particular, $m\ge n-1$.

Let $S$ be a random cut obtained by sampling every vertex $v\in V$ independently with probability $1/2$. Every edge $e\in E$ crosses $S$ with probability exactly $1/2$, so $\mexp{C_{S,G}}=m/2$. 

\textbf{Case 1: $\Delta<c\log n$.} In this case we have $\Delta < \eps m/2$, since $m\ge n-1\ge (2c/\eps)\log n$, for large enough $n$. Thus, we have 
$\mexp{\varphi(S)}=\mexp{C_{S,G}-\max_{v\in V}d_S(v)}
= \mexp{C_{S,G}}-\mexp{\max_{v\in V}d_S(v)}\ge m/2 - \Delta\ge (1-\eps)m/2$, hence we get a $(1-\eps)/2$-approximation.

\textbf{Case 2: $\Delta\ge c\log n$.} Note that for every vertex $v\in V$, $\mexp{d_S(v)}=d(v)/2$. Thus, for $v$ with $d(v)>\Delta/2$, we have $\mexp{d_S(v)}>\Delta/4\ge (c/4)\log n$. By a Chernoff bound (note that $d_S(v)$ is a sum of independent Bernoulli random variables, one for each adjacent edge of $v$) and the choice of $c$, we have:
\begin{align*}\mpr{|d_S(v)-d(v)/2|\ge \eps^2\cdot (d(v)/2)} &\le 2\cdot \exp(-(\eps^4/3) \mexp{d_S(v)})\\
&\le 2\cdot\exp(-(\eps^4 c/12)\log n)\le 2\cdot n^{-4}\ .
\end{align*}
 Let $\mathcal{E}$ be the event that  $d_S(v)=(1\pm\eps^2)d(v)/2$ holds for all $v\in V$ with $d(v)\ge \Delta/2$. By the union bound, we have that $\mpr{\mathcal{E}}\ge 1-n^{-2}$. We consider two cases.

\textbf{Case 2.1: $m>(1+2\eps)\Delta$}.
Recall that $\mexp{C_{S,G}}=m/2$. Thus,
\[
m/2=\mexp{C_{S,G}}=\mexp{C_{S,G} \mid \mathcal{E}}\mpr{\mathcal{E}} + \mexp{C_{S,G} \mid \neg \mathcal{E}}\mpr{\neg\mathcal{E}}\le \mexp{C_{S,G} \mid \mathcal{E}}\mpr{\mathcal{E}} + m/n^2\ ,
\]
where the last inequality follows because $\mpr{\neg \mathcal{E}} \leq 1/n^2$, and $C_{S,G}$ is bounded by $m$. Hence, $\mexp{C_{S,G} \mid \mathcal{E}}\mpr{\mathcal{E}}\ge m/2-m/n^2\ge m/2 -1$. It then follows that
\begin{align*}
\mexp{\varphi(S)}&\ge \mexp{C_{S,G}-\max_v{d_S(v)}}\\
&\ge \mexp{C_{S,G}-\max_v{d_S(v)}\mid \mathcal{E}}\mpr{\mathcal{E}}\\
&=\mexp{C_{S,G}\mid \mathcal{E}}\mpr{\mathcal{E}} - \mexp{\max_v{d_S(v)}\mid \mathcal{E}}\mpr{\mathcal{E}}\\
&\ge m/2-1 - (1+\eps^2)\Delta/2\\
&=(m-\Delta)/2 - 1 - \eps^2\Delta/2\ ,
\end{align*}
where the last inequality follows because given $\mathcal{E}$ the maximal degree is $(1+\epsilon^2)\Delta/2$. Using $2\eps\Delta \le m-\Delta$ and $\Delta \geq c\log n$, we get that $1<\eps^2\Delta/2\le (\eps/2)\cdot (m-\Delta)/2$; hence, $1 + \eps^2\Delta/2 \leq \eps(m-\Delta)/2$, and we get the claimed $(1-\epsilon)/2$-approximation:
\[
\mexp{\varphi(S)}\ge (m-\Delta)/2 - \eps(m-\Delta)/2 \ge (1-\eps)(m-\Delta)/2\ .
\]

\textbf{Case 2.2:
$m\le (1+2\epsilon)\Delta$.} It follows that $m < 2\Delta - 1$, therefore there is a unique vertex $v$ with $d(v)=\Delta$. Let $E'$ be the set of edges not adjacent to $v$. By the assumption, $|E'|\le 2\epsilon\Delta$; hence, the set $V'$ of vertices adjacent to an edge in $E'$ has size  $|V'|\le 2|E'|\le 4\eps\Delta$. Note that for every $w\notin V'$, $d(w) = 1$. Let $\mathcal{E}'$ be the event that $X=d_{S-V',G-V'}(v)>2\epsilon\Delta$ (note that $X$ is a sum of independent Bernoulli random variables, one for each adjacent edge of $v$ in $G-V'$).  Since $d_{G-V'}(v)>(1-4\epsilon)\Delta$, $\mexp{X} > (1-4\epsilon)\Delta/2$. First we show that $\mpr{\mathcal{E}'}\ge 1-n^{-2}$. Using $\epsilon<1/16$ and a Chernoff bound, we have that 
\begin{align*}
    \mpr{X \le 2\epsilon\Delta}\le &\mpr{X \le \frac{4\epsilon}{1-4\epsilon}\mexp{X}} \\
    \le &\mpr{X \le (1/3)\cdot \mexp{X}} \\
    \le& \exp\left(-(2/3)^2\mexp{X}/2\right)\\ 
    \le & \exp\left(-(1-4\epsilon)\Delta/9)\right)\le \exp(-\Delta/12)\ .
\end{align*}
Using $\Delta\ge c\log n$ with $c\ge 24$, we get that $\mpr{X \le 2\epsilon\Delta} \leq n^{-2}$, thus $\mpr{\mathcal{E}'} \ge 1-n^{-2}$.
Note that given $\mathcal{E'}$, $d_{S,G}(v)>2\epsilon\Delta$, while $d_{S,G}(u)\le 2\epsilon\Delta$, for every $u\in V-v$; hence $v$ is a critical vertex for $S$, i.e., 
$\varphi(S)=C_{S-v,G-v}$. Hence, we have:
\begin{align*}
\mexp{\varphi(S)}\ge \mexp{\varphi(S)\mid \mathcal{E}'}\cdot \mpr{\mathcal{E}'}&=\mexp{C_{S-v,G-v}\mid \mathcal{E}'}\cdot \mpr{\mathcal{E}'}\\
&=\mexp{C_{S-v,G-v}}\cdot \mpr{\mathcal{E}'}\\
&\ge(1-n^{-2})\cdot (m-\Delta)/2\ge (1-\epsilon)(m-\Delta)/2\ ,
\end{align*}
where to get the second row, we use the fact that $C_{S-v,G-v}$ is independent of $\mathcal{E}'$, since $\mathcal{E}'$ only conditions on vertices that can never contribute to the cut $S-v$ in $G-v$, because all of their edges are adjacent to $v$. This completes the proof.
\end{proof}

\begin{remark}Note that as opposed to \mc, a random cut gives a $(1/2-\epsilon)$-approximation for large $n$ only. For example, we show that for an unweighted $4$-cycle, it holds that $\mexp{\varphi(S)} = 1/4\cdot \varphi(S^*)$, where  $S^*$ is an optimal \ftc.

Let $v_1,v_2,v_3,v_4,v_1$ be the 4-cycle. Overall, it has 16 cuts, so the probability of each cut being output by the algorithm is $1/16$. For every cut $S$ such that $|S|\in \set{0,1,3,4}$ (there are 10 such cuts), it holds that $\varphi(S) = 0$, since when $|S|\in \set{0,4}$ it holds that $C_S = 0$, and when $|S|\in \set{1, 3}$, all the crossing edges are adjacent to one vertex, therefore when that vertex fails, no crossing edges remain. For $|S|=2$, we have $6$ options. When $S$ consists of two adjacent vertices (there are four options), it holds that $\varphi(S) = 1$. When $S$ is $\set{v_1, v_3}$ or $\set{v_2, v_4}$, it holds that $\varphi(S) = 2$.
Altogether, we get that\[\mexp{\varphi(S)} = 0\cdot 10/16 + 1\cdot 4/16 + 2\cdot 2/16 = 1/2.\]
However, an optimal \ftc is for example $S^* = \set{v_1, v_3}$, for which $\varphi(S^*) = 2$; hence, $\mexp{S} = 1/4\cdot \varphi(S^*)$.
\end{remark}

\subsection{Oblivious Adversary, Weighted Instances, and Many Faults}

Here we show that the uniform distribution (output every cut with probability $1/2^n$) is a $1/2$-approximation for \koblftc, for every $k$ (not necessarily constant). 

\begin{theorem}\label{thm:obl_random_cut_k}
Let $G=(V,E,w)$ be a weighted graph and $k\in \mathbb{N}$ be a number. Let $\mathcal{U}$ be the uniform distribution over all cuts in $G$. Then, $\mathcal{U}$ is a $1/2$-approximation for \kftc against an oblivious adversary. 
\end{theorem}

Let $\Delta_k$ denote the maximum degree (total weight of adjacent edges) of a subset of size $k$, that is, $\Delta_k = \max_{F\in \binom{V}{k}}d(F)$.

Before proving the theorem, we make two observations: 1. failing a subset $F$ removes $d_S(F)$ edges from a cut, and 2. there is a subset $F$ whose removal always limits the remaining cut size by $m-\Delta_k$. 
 
\begin{observation}\label{obs:obl_adversary_removal_k}
For every cut $S\subseteq V$ and subset $F\in \binom{V}{k}$, it holds that $C_{S-F,G-F}=C_S-d_S(F)$.
\end{observation}

\begin{observation}\label{obs:best_dist_obl_k}
There is a subset $F\in \binom{V}{k}$ such that for every distribution $\mathcal{D}$, it holds that $\mexpl{{S \sim \mathcal{D}}}{C_{S\sm F, G\sm F}} \leq m-\Delta_k.$
\end{observation}

\begin{proof}
For every cut $S\subseteq V$ and a subset $F\in \binom{V}{k}$ with $d(F)=\Delta_k$, it holds that $C_{S-F,G-F}\le m-\Delta_k$, since the right side is the total edge-weight in $G-F$. Thus, $\mexpl{{S \sim \mathcal{D}}}{C_{S\sm F, G\sm F}} \leq\max_S C_{S\sm F, G\sm F}\le m-\Delta_k$, as claimed.
\end{proof}

\begin{proof}[Proof of Theorem~\ref{thm:obl_random_cut_k}] Let $F\in \binom{V}{k}$ be  arbitrary. Note that for $S\sim\mathcal{U}$, every edge $e\in E$ crosses $S$ with probability $1/2$; hence, by linearity of expectation, it holds that $\mexpl{{S \sim \mathcal{U}}}{d_S(F)} = d(F)/2$, and $\mexpl{{S \sim \mathcal{U}}}{C_S} = m/2$. 
By Observation~\ref{obs:obl_adversary_removal_k}, we have 
\[
\mexpl{{S \sim \mathcal{U}}}{C_{S\sm F, G\sm F}} =\mexpl{{S \sim \mathcal{U}}}{C_{S}}-\mexpl{{S \sim \mathcal{U}}}{d_{S}(F)}=(m-d(F))/2\ge  (m-\Delta_k)/2\ ,
\]
since, by definition, $d(F)\le \Delta_k$. Since this holds for every $F\in \binom{V}{k}$, Observation~\ref{obs:best_dist_obl_k} implies that $\mathcal{U}$ is a $1/2$-approximation, as claimed.
\end{proof}

\section{Discussion}
Our work leaves several open questions regarding fault tolerant \mc.
An immediate question is to bridge the (rather small) gap between our approximation of $(0.8780-\epsilon)$ and our hardness of $\agw$ for \kftc.

The central bottleneck is that \smc, a main ingredient in our algorithm, has  hardness of approximation that is slightly below $\agw$ and equals $ (\agw - \delta)$ (where $ \delta \geq 10^{-5}$) \cite{DBLP:conf/coco/BhangaleK20}.
Thus, either one finds a different algorithm for \kftc that does not rely on \smc and achieves an approximation of $\agw$, or one can extend the hardness result of \cite{DBLP:conf/coco/BhangaleK20} to \kftc and thus rule out an approximation of $\agw$ for \kftc. Another question is what approximation factors can be obtained for  \ftc on general weighted graphs. 

Another interesting question is how to deal with a non-constant number of faults, for both of the adversaries. Since the number of all possible cases of failure is not polynomial, a new approach may be needed. There are techniques that are used to deal with a non-constant number of faults, e.g., failure sampling, that is presented in \cite{DBLP:conf/stoc/DinitzK11}. It would be interesting to see whether these techniques can be used for fault tolerant \mc as well.

One more important and intriguing open question is what happens in other fault tolerant problems when an oblivious adversary is considered. We are unaware of previous algorithms for an oblivious adversary in the fault-tolerance literature. Since an oblivious adversary is arguably more realistic in its nature, and since it is likely that one can get improved algorithms for this case, pursuing this line of research could be crucial for many additional fundamental  problems involving fault tolerance.


\begin{thebibliography}{10}

\bibitem{DBLP:conf/icalp/AlonCC19}
Noga Alon, Shiri Chechik, and Sarel Cohen.
\newblock Deterministic combinatorial replacement paths and distance
  sensitivity oracles.
\newblock In Christel Baier, Ioannis Chatzigiannakis, Paola Flocchini, and
  Stefano Leonardi, editors, {\em 46th International Colloquium on Automata,
  Languages, and Programming, {ICALP} 2019, July 9-12, 2019, Patras, Greece},
  volume 132 of {\em LIPIcs}, pages 12:1--12:14. Schloss Dagstuhl -
  Leibniz-Zentrum f{\"{u}}r Informatik, 2019.
\newblock \href {https://doi.org/10.4230/LIPIcs.ICALP.2019.12}
  {\path{doi:10.4230/LIPIcs.ICALP.2019.12}}.

\bibitem{AlonNaor}
Noga Alon and Assaf Naor.
\newblock Approximating the cut-norm via grothendieck's inequality.
\newblock {\em {SIAM} J. Comput.}, 35(4):787--803, 2006.

\bibitem{DBLP:journals/dam/AngelBG06}
Eric Angel, Evripidis Bampis, and Laurent Gourv{\`{e}}s.
\newblock Approximation algorithms for the bi-criteria weighted {MAX-CUT}
  problem.
\newblock {\em Discret. Appl. Math.}, 154(12):1685--1692, 2006.
\newblock \href {https://doi.org/10.1016/j.dam.2006.02.008}
  {\path{doi:10.1016/j.dam.2006.02.008}}.

\bibitem{asano2002improved}
Takao Asano and David~P Williamson.
\newblock Improved approximation algorithms for {MAX} {SAT}.
\newblock {\em Journal of Algorithms}, 42(1):173--202, 2002.

\bibitem{AustrinBG16}
Per Austrin, Siavosh Benabbas, and Konstantinos Georgiou.
\newblock Better balance by being biased: {A} 0.8776-approximation for {Max
  Bisection}.
\newblock {\em {ACM} Trans. Algorithms}, 13(1):2:1--2:27, 2016.

\bibitem{DBLP:conf/kdd/AvdiukhinMYZ19}
Dmitrii Avdiukhin, Slobodan Mitrovic, Grigory Yaroslavtsev, and Samson Zhou.
\newblock Adversarially robust submodular maximization under knapsack
  constraints.
\newblock In Ankur Teredesai, Vipin Kumar, Ying Li, R{\'{o}}mer Rosales,
  Evimaria Terzi, and George Karypis, editors, {\em Proceedings of the 25th
  {ACM} {SIGKDD} International Conference on Knowledge Discovery {\&} Data
  Mining, {KDD} 2019, Anchorage, AK, USA, August 4-8, 2019}, pages 148--156.
  {ACM}, 2019.
\newblock \href {https://doi.org/10.1145/3292500.3330911}
  {\path{doi:10.1145/3292500.3330911}}.

\bibitem{DBLP:conf/waoa/AvidorBZ05}
Adi Avidor, Ido Berkovitch, and Uri Zwick.
\newblock Improved approximation algorithms for {MAX} {NAE-SAT} and {MAX}
  {SAT}.
\newblock In Thomas Erlebach and Giuseppe Persiano, editors, {\em Approximation
  and Online Algorithms, Third International Workshop, {WAOA} 2005, Palma de
  Mallorca, Spain, October 6-7, 2005, Revised Papers}, volume 3879 of {\em
  Lecture Notes in Computer Science}, pages 27--40. Springer, 2005.
\newblock \href {https://doi.org/10.1007/11671411\_3}
  {\path{doi:10.1007/11671411\_3}}.

\bibitem{DBLP:conf/soda/BaswanaCHR18}
Surender Baswana, Keerti Choudhary, Moazzam Hussain, and Liam Roditty.
\newblock Approximate single source fault tolerant shortest path.
\newblock In Artur Czumaj, editor, {\em Proceedings of the Twenty-Ninth Annual
  {ACM-SIAM} Symposium on Discrete Algorithms, {SODA} 2018, New Orleans, LA,
  USA, January 7-10, 2018}, pages 1901--1915. {SIAM}, 2018.
\newblock \href {https://doi.org/10.1137/1.9781611975031.124}
  {\path{doi:10.1137/1.9781611975031.124}}.

\bibitem{DBLP:conf/wdag/BaswanaCR15}
Surender Baswana, Keerti Choudhary, and Liam Roditty.
\newblock Fault tolerant reachability for directed graphs.
\newblock In Yoram Moses, editor, {\em Distributed Computing - 29th
  International Symposium, {DISC} 2015, Tokyo, Japan, October 7-9, 2015,
  Proceedings}, volume 9363 of {\em Lecture Notes in Computer Science}, pages
  528--543. Springer, 2015.
\newblock \href {https://doi.org/10.1007/978-3-662-48653-5\_35}
  {\path{doi:10.1007/978-3-662-48653-5\_35}}.

\bibitem{DBLP:conf/stoc/BaswanaCR16}
Surender Baswana, Keerti Choudhary, and Liam Roditty.
\newblock Fault tolerant subgraph for single source reachability: generic and
  optimal.
\newblock In Daniel Wichs and Yishay Mansour, editors, {\em Proceedings of the
  48th Annual {ACM} {SIGACT} Symposium on Theory of Computing, {STOC} 2016,
  Cambridge, MA, USA, June 18-21, 2016}, pages 509--518. {ACM}, 2016.
\newblock \href {https://doi.org/10.1145/2897518.2897648}
  {\path{doi:10.1145/2897518.2897648}}.

\bibitem{DBLP:conf/coco/BhangaleK20}
Amey Bhangale and Subhash Khot.
\newblock Simultaneous max-cut is harder to approximate than max-cut.
\newblock In Shubhangi Saraf, editor, {\em 35th Computational Complexity
  Conference, {CCC} 2020, July 28-31, 2020, Saarbr{\"{u}}cken, Germany (Virtual
  Conference)}, volume 169 of {\em LIPIcs}, pages 9:1--9:15. Schloss Dagstuhl -
  Leibniz-Zentrum f{\"{u}}r Informatik, 2020.
\newblock \href {https://doi.org/10.4230/LIPIcs.CCC.2020.9}
  {\path{doi:10.4230/LIPIcs.CCC.2020.9}}.

\bibitem{DBLP:conf/soda/BhangaleKKST18}
Amey Bhangale, Subhash Khot, Swastik Kopparty, Sushant Sachdeva, and Devanathan
  Thiruvenkatachari.
\newblock Near-optimal approximation algorithm for simultaneous max-cut.
\newblock In Artur Czumaj, editor, {\em Proceedings of the Twenty-Ninth Annual
  {ACM-SIAM} Symposium on Discrete Algorithms, {SODA} 2018, New Orleans, LA,
  USA, January 7-10, 2018}, pages 1407--1425. {SIAM}, 2018.
\newblock \href {https://doi.org/10.1137/1.9781611975031.93}
  {\path{doi:10.1137/1.9781611975031.93}}.

\bibitem{DBLP:conf/icalp/BhangaleKS15}
Amey Bhangale, Swastik Kopparty, and Sushant Sachdeva.
\newblock Simultaneous approximation of constraint satisfaction problems.
\newblock In Magn{\'{u}}s~M. Halld{\'{o}}rsson, Kazuo Iwama, Naoki Kobayashi,
  and Bettina Speckmann, editors, {\em Automata, Languages, and Programming -
  42nd International Colloquium, {ICALP} 2015, Kyoto, Japan, July 6-10, 2015,
  Proceedings, Part {I}}, volume 9134 of {\em Lecture Notes in Computer
  Science}, pages 193--205. Springer, 2015.
\newblock \href {https://doi.org/10.1007/978-3-662-47672-7\_16}
  {\path{doi:10.1007/978-3-662-47672-7\_16}}.

\bibitem{DBLP:journals/algorithmica/BiloGLP18}
Davide Bil{\`{o}}, Luciano Gual{\`{a}}, Stefano Leucci, and Guido Proietti.
\newblock Fault-tolerant approximate shortest-path trees.
\newblock {\em Algorithmica}, 80(12):3437--3460, 2018.
\newblock \href {https://doi.org/10.1007/s00453-017-0396-z}
  {\path{doi:10.1007/s00453-017-0396-z}}.

\bibitem{DBLP:conf/soda/BodwinDPW18}
Greg Bodwin, Michael Dinitz, Merav Parter, and Virginia~Vassilevska Williams.
\newblock Optimal vertex fault tolerant spanners (for fixed stretch).
\newblock In Artur Czumaj, editor, {\em Proceedings of the Twenty-Ninth Annual
  {ACM-SIAM} Symposium on Discrete Algorithms, {SODA} 2018, New Orleans, LA,
  USA, January 7-10, 2018}, pages 1884--1900. {SIAM}, 2018.
\newblock \href {https://doi.org/10.1137/1.9781611975031.123}
  {\path{doi:10.1137/1.9781611975031.123}}.

\bibitem{DBLP:conf/icalp/BodwinGPW17}
Greg Bodwin, Fabrizio Grandoni, Merav Parter, and Virginia~Vassilevska
  Williams.
\newblock Preserving distances in very faulty graphs.
\newblock In Chatzigiannakis et~al. \cite{DBLP:conf/icalp/2017}, pages
  73:1--73:14.
\newblock \href {https://doi.org/10.4230/LIPIcs.ICALP.2017.73}
  {\path{doi:10.4230/LIPIcs.ICALP.2017.73}}.

\bibitem{DBLP:conf/podc/BodwinP19}
Greg Bodwin and Shyamal Patel.
\newblock A trivial yet optimal solution to vertex fault tolerant spanners.
\newblock In Peter Robinson and Faith Ellen, editors, {\em Proceedings of the
  2019 {ACM} Symposium on Principles of Distributed Computing, {PODC} 2019,
  Toronto, ON, Canada, July 29 - August 2, 2019}, pages 541--543. {ACM}, 2019.
\newblock \href {https://doi.org/10.1145/3293611.3331588}
  {\path{doi:10.1145/3293611.3331588}}.

\bibitem{DBLP:journals/informs/BuchananSBP15}
Austin Buchanan, Je~Sang Sung, Sergiy Butenko, and Eduardo~L. Pasiliao.
\newblock An integer programming approach for fault-tolerant connected
  dominating sets.
\newblock {\em {INFORMS} J. Comput.}, 27(1):178--188, 2015.
\newblock \href {https://doi.org/10.1287/ijoc.2014.0619}
  {\path{doi:10.1287/ijoc.2014.0619}}.

\bibitem{CharikarGW05}
Moses Charikar, Venkatesan Guruswami, and Anthony Wirth.
\newblock Clustering with qualitative information.
\newblock {\em J. Comput. Syst. Sci.}, 71(3):360--383, 2005.

\bibitem{DBLP:conf/icalp/2017}
Ioannis Chatzigiannakis, Piotr Indyk, Fabian Kuhn, and Anca Muscholl, editors.
\newblock {\em 44th International Colloquium on Automata, Languages, and
  Programming, {ICALP} 2017, July 10-14, 2017, Warsaw, Poland}, volume~80 of
  {\em LIPIcs}. Schloss Dagstuhl - Leibniz-Zentrum f{\"{u}}r Informatik, 2017.
\newblock URL: \url{http://www.dagstuhl.de/dagpub/978-3-95977-041-5}.

\bibitem{DBLP:conf/soda/ChechikC19}
Shiri Chechik and Sarel Cohen.
\newblock Near optimal algorithms for the single source replacement paths
  problem.
\newblock In Timothy~M. Chan, editor, {\em Proceedings of the Thirtieth Annual
  {ACM-SIAM} Symposium on Discrete Algorithms, {SODA} 2019, San Diego,
  California, USA, January 6-9, 2019}, pages 2090--2109. {SIAM}, 2019.
\newblock \href {https://doi.org/10.1137/1.9781611975482.126}
  {\path{doi:10.1137/1.9781611975482.126}}.

\bibitem{DBLP:conf/icalp/ChechikM20}
Shiri Chechik and Ofer Magen.
\newblock Near optimal algorithm for the directed single source replacement
  paths problem.
\newblock In Artur Czumaj, Anuj Dawar, and Emanuela Merelli, editors, {\em 47th
  International Colloquium on Automata, Languages, and Programming, {ICALP}
  2020, July 8-11, 2020, Saarbr{\"{u}}cken, Germany (Virtual Conference)},
  volume 168 of {\em LIPIcs}, pages 81:1--81:17. Schloss Dagstuhl -
  Leibniz-Zentrum f{\"{u}}r Informatik, 2020.
\newblock \href {https://doi.org/10.4230/LIPIcs.ICALP.2020.81}
  {\path{doi:10.4230/LIPIcs.ICALP.2020.81}}.

\bibitem{DBLP:journals/tcs/ChechikP14}
Shiri Chechik and David Peleg.
\newblock Robust fault tolerant uncapacitated facility location.
\newblock {\em Theor. Comput. Sci.}, 543:9--23, 2014.
\newblock \href {https://doi.org/10.1016/j.tcs.2014.05.013}
  {\path{doi:10.1016/j.tcs.2014.05.013}}.

\bibitem{DBLP:conf/stoc/DinitzK11}
Michael Dinitz and Robert Krauthgamer.
\newblock Directed spanners via flow-based linear programs.
\newblock In Lance Fortnow and Salil~P. Vadhan, editors, {\em Proceedings of
  the 43rd {ACM} Symposium on Theory of Computing, {STOC} 2011, San Jose, CA,
  USA, 6-8 June 2011}, pages 323--332. {ACM}, 2011.
\newblock \href {https://doi.org/10.1145/1993636.1993680}
  {\path{doi:10.1145/1993636.1993680}}.

\bibitem{DBLP:conf/podc/DinitzK11}
Michael Dinitz and Robert Krauthgamer.
\newblock Fault-tolerant spanners: better and simpler.
\newblock In Cyril Gavoille and Pierre Fraigniaud, editors, {\em Proceedings of
  the 30th Annual {ACM} Symposium on Principles of Distributed Computing,
  {PODC} 2011, San Jose, CA, USA, June 6-8, 2011}, pages 169--178. {ACM}, 2011.
\newblock \href {https://doi.org/10.1145/1993806.1993830}
  {\path{doi:10.1145/1993806.1993830}}.

\bibitem{feige1995approximating}
Uriel Feige and Michel Goemans.
\newblock Approximating the value of two power proof systems, with applications
  to {M}ax 2-{SAT} and {Max} {D}i{C}ut.
\newblock In {\em istcs}, page 0182. IEEE, 1995.

\bibitem{DBLP:journals/tcs/GareyJS76}
M.~R. Garey, David~S. Johnson, and Larry~J. Stockmeyer.
\newblock Some simplified np-complete graph problems.
\newblock {\em Theor. Comput. Sci.}, 1(3):237--267, 1976.
\newblock \href {https://doi.org/10.1016/0304-3975(76)90059-1}
  {\path{doi:10.1016/0304-3975(76)90059-1}}.

\bibitem{DBLP:conf/spaa/GhaffariP16}
Mohsen Ghaffari and Merav Parter.
\newblock Near-optimal distributed algorithms for fault-tolerant tree
  structures.
\newblock In Christian Scheideler and Seth Gilbert, editors, {\em Proceedings
  of the 28th {ACM} Symposium on Parallelism in Algorithms and Architectures,
  {SPAA} 2016, Asilomar State Beach/Pacific Grove, CA, USA, July 11-13, 2016},
  pages 387--396. {ACM}, 2016.
\newblock \href {https://doi.org/10.1145/2935764.2935795}
  {\path{doi:10.1145/2935764.2935795}}.

\bibitem{DBLP:journals/jacm/GoemansW95}
Michel~X. Goemans and David~P. Williamson.
\newblock Improved approximation algorithms for maximum cut and satisfiability
  problems using semidefinite programming.
\newblock {\em J. {ACM}}, 42(6):1115--1145, 1995.
\newblock \href {https://doi.org/10.1145/227683.227684}
  {\path{doi:10.1145/227683.227684}}.

\bibitem{DBLP:conf/focs/GrandoniW12}
Fabrizio Grandoni and Virginia~Vassilevska Williams.
\newblock Improved distance sensitivity oracles via fast single-source
  replacement paths.
\newblock In {\em 53rd Annual {IEEE} Symposium on Foundations of Computer
  Science, {FOCS} 2012, New Brunswick, NJ, USA, October 20-23, 2012}, pages
  748--757. {IEEE} Computer Society, 2012.
\newblock \href {https://doi.org/10.1109/FOCS.2012.17}
  {\path{doi:10.1109/FOCS.2012.17}}.

\bibitem{ellipsoid}
Martin Gr{\"o}tschel, L{\'a}szl{\'o} Lov{\'a}sz, and Alexander Schrijver.
\newblock {\em The Ellipsoid Method}, pages 64--101.
\newblock Springer Berlin Heidelberg, Berlin, Heidelberg, 1993.
\newblock \href {https://doi.org/10.1007/978-3-642-78240-4_4}
  {\path{doi:10.1007/978-3-642-78240-4_4}}.

\bibitem{DBLP:journals/jal/GuhaMM03}
Sudipto Guha, Adam Meyerson, and Kamesh Munagala.
\newblock A constant factor approximation algorithm for the fault-tolerant
  facility location problem.
\newblock {\em J. Algorithms}, 48(2):429--440, 2003.
\newblock \href {https://doi.org/10.1016/S0196-6774(03)00056-7}
  {\path{doi:10.1016/S0196-6774(03)00056-7}}.

\bibitem{DBLP:conf/icalp/GuptaK17}
Manoj Gupta and Shahbaz Khan.
\newblock Multiple source dual fault tolerant {BFS} trees.
\newblock In Chatzigiannakis et~al. \cite{DBLP:conf/icalp/2017}, pages
  127:1--127:15.
\newblock \href {https://doi.org/10.4230/LIPIcs.ICALP.2017.127}
  {\path{doi:10.4230/LIPIcs.ICALP.2017.127}}.

\bibitem{DBLP:journals/jacm/Hastad01}
Johan H{\aa}stad.
\newblock Some optimal inapproximability results.
\newblock {\em J. {ACM}}, 48(4):798--859, 2001.
\newblock \href {https://doi.org/10.1145/502090.502098}
  {\path{doi:10.1145/502090.502098}}.

\bibitem{DBLP:conf/soda/JainMS03}
Kamal Jain, Mohammad Mahdian, and Mohammad~R. Salavatipour.
\newblock Packing steiner trees.
\newblock In {\em Proceedings of the Fourteenth Annual {ACM-SIAM} Symposium on
  Discrete Algorithms, January 12-14, 2003, Baltimore, Maryland, {USA}}, pages
  266--274. {ACM/SIAM}, 2003.
\newblock URL: \url{http://dl.acm.org/citation.cfm?id=644108.644154}.

\bibitem{DBLP:journals/algorithmica/JainV03}
Kamal Jain and Vijay~V. Vazirani.
\newblock An approximation algorithm for the fault tolerant metric facility
  location problem.
\newblock {\em Algorithmica}, 38(3):433--439, 2004.
\newblock \href {https://doi.org/10.1007/s00453-003-1070-1}
  {\path{doi:10.1007/s00453-003-1070-1}}.

\bibitem{Karp_np_complete}
Richard~M. Karp.
\newblock Reducibility among combinatorial problems.
\newblock In Raymond~E. Miller and James~W. Thatcher, editors, {\em Proceedings
  of a symposium on the Complexity of Computer Computations, held March 20-22,
  1972, at the {IBM} Thomas J. Watson Research Center, Yorktown Heights, New
  York, {USA}}, The {IBM} Research Symposia Series, pages 85--103. Plenum
  Press, New York, 1972.
\newblock \href {https://doi.org/10.1007/978-1-4684-2001-2\_9}
  {\path{doi:10.1007/978-1-4684-2001-2\_9}}.

\bibitem{DBLP:conf/stoc/Khot02a}
Subhash Khot.
\newblock On the power of unique 2-prover 1-round games.
\newblock In John~H. Reif, editor, {\em Proceedings on 34th Annual {ACM}
  Symposium on Theory of Computing, May 19-21, 2002, Montr{\'{e}}al,
  Qu{\'{e}}bec, Canada}, pages 767--775. {ACM}, 2002.
\newblock \href {https://doi.org/10.1145/509907.510017}
  {\path{doi:10.1145/509907.510017}}.

\bibitem{DBLP:journals/siamcomp/KhotKMO07}
Subhash Khot, Guy Kindler, Elchanan Mossel, and Ryan O'Donnell.
\newblock Optimal inapproximability results for {MAX-CUT} and other 2-variable
  csps?
\newblock {\em {SIAM} J. Comput.}, 37(1):319--357, 2007.
\newblock \href {https://doi.org/10.1137/S0097539705447372}
  {\path{doi:10.1137/S0097539705447372}}.

\bibitem{DBLP:journals/jcss/KhotR08}
Subhash Khot and Oded Regev.
\newblock Vertex cover might be hard to approximate to within 2-$\epsilon$.
\newblock {\em J. Comput. Syst. Sci.}, 74(3):335--349, 2008.
\newblock \href {https://doi.org/10.1016/j.jcss.2007.06.019}
  {\path{doi:10.1016/j.jcss.2007.06.019}}.

\bibitem{DBLP:journals/algorithmica/LevcopoulosNS02}
Christos Levcopoulos, Giri Narasimhan, and Michiel H.~M. Smid.
\newblock Improved algorithms for constructing fault-tolerant spanners.
\newblock {\em Algorithmica}, 32(1):144--156, 2002.
\newblock \href {https://doi.org/10.1007/s00453-001-0075-x}
  {\path{doi:10.1007/s00453-001-0075-x}}.

\bibitem{DBLP:conf/ipco/LewinLZ02}
Michael Lewin, Dror Livnat, and Uri Zwick.
\newblock Improved rounding techniques for the {MAX} 2-sat and {MAX} {DI-CUT}
  problems.
\newblock In William~J. Cook and Andreas~S. Schulz, editors, {\em Integer
  Programming and Combinatorial Optimization, 9th International {IPCO}
  Conference, Cambridge, MA, USA, May 27-29, 2002, Proceedings}, volume 2337 of
  {\em Lecture Notes in Computer Science}, pages 67--82. Springer, 2002.
\newblock \href {https://doi.org/10.1007/3-540-47867-1\_6}
  {\path{doi:10.1007/3-540-47867-1\_6}}.

\bibitem{DBLP:journals/siamcomp/MahajanR99}
Sanjeev Mahajan and H.~Ramesh.
\newblock Derandomizing approximation algorithms based on semidefinite
  programming.
\newblock {\em {SIAM} J. Comput.}, 28(5):1641--1663, 1999.
\newblock \href {https://doi.org/10.1137/S0097539796309326}
  {\path{doi:10.1137/S0097539796309326}}.

\bibitem{DBLP:conf/random/MatuuraM01}
Shiro Matuura and Tomomi Matsui.
\newblock 63-approximation algorithm for {MAX} {DICUT}.
\newblock In Michel~X. Goemans, Klaus Jansen, Jos{\'{e}} D.~P. Rolim, and Luca
  Trevisan, editors, {\em Approximation, Randomization and Combinatorial
  Optimization: Algorithms and Techniques, 4th International Workshop on
  Approximation Algorithms for Combinatorial Optimization Problems, {APPROX}
  2001 and 5th International Workshop on Randomization and Approximation
  Techniques in Computer Science, {RANDOM} 2001 Berkeley, CA, USA, August
  18-20, 2001, Proceedings}, volume 2129 of {\em Lecture Notes in Computer
  Science}, pages 138--146. Springer, 2001.
\newblock \href {https://doi.org/10.1007/3-540-44666-4\_17}
  {\path{doi:10.1007/3-540-44666-4\_17}}.

\bibitem{MOO}
Elchanan Mossel, Ryan O'Donnell, and Krzysztof Oleszkiewicz.
\newblock Noise stability of functions with low influences: invariance and
  optimality.
\newblock In {\em Foundations of Computer Science, {(FOCS)}}, pages 21--30,
  2005.

\bibitem{DBLP:conf/ipco/OrlinSU16}
James~B. Orlin, Andreas~S. Schulz, and Rajan Udwani.
\newblock Robust monotone submodular function maximization.
\newblock In Quentin Louveaux and Martin Skutella, editors, {\em Integer
  Programming and Combinatorial Optimization - 18th International Conference,
  {IPCO} 2016, Li{\`{e}}ge, Belgium, June 1-3, 2016, Proceedings}, volume 9682
  of {\em Lecture Notes in Computer Science}, pages 312--324. Springer, 2016.
\newblock \href {https://doi.org/10.1007/978-3-319-33461-5\_26}
  {\path{doi:10.1007/978-3-319-33461-5\_26}}.

\bibitem{DBLP:conf/wdag/Parter14}
Merav Parter.
\newblock Vertex fault tolerant additive spanners.
\newblock In Fabian Kuhn, editor, {\em Distributed Computing - 28th
  International Symposium, {DISC} 2014, Austin, TX, USA, October 12-15, 2014.
  Proceedings}, volume 8784 of {\em Lecture Notes in Computer Science}, pages
  167--181. Springer, 2014.
\newblock \href {https://doi.org/10.1007/978-3-662-45174-8\_12}
  {\path{doi:10.1007/978-3-662-45174-8\_12}}.

\bibitem{DBLP:conf/podc/Parter15}
Merav Parter.
\newblock Dual failure resilient {BFS} structure.
\newblock In Chryssis Georgiou and Paul~G. Spirakis, editors, {\em Proceedings
  of the 2015 {ACM} Symposium on Principles of Distributed Computing, {PODC}
  2015, Donostia-San Sebasti{\'{a}}n, Spain, July 21 - 23, 2015}, pages
  481--490. {ACM}, 2015.
\newblock \href {https://doi.org/10.1145/2767386.2767408}
  {\path{doi:10.1145/2767386.2767408}}.

\bibitem{DBLP:conf/wdag/Parter20}
Merav Parter.
\newblock Distributed constructions of dual-failure fault-tolerant distance
  preservers.
\newblock In Hagit Attiya, editor, {\em 34th International Symposium on
  Distributed Computing, {DISC} 2020, October 12-16, 2020, Virtual Conference},
  volume 179 of {\em LIPIcs}, pages 21:1--21:17. Schloss Dagstuhl -
  Leibniz-Zentrum f{\"{u}}r Informatik, 2020.
\newblock \href {https://doi.org/10.4230/LIPIcs.DISC.2020.21}
  {\path{doi:10.4230/LIPIcs.DISC.2020.21}}.

\bibitem{DBLP:journals/corr/abs-1302-5401}
Merav Parter and David Peleg.
\newblock Sparse fault-tolerant {BFS} trees.
\newblock {\em CoRR}, abs/1302.5401, 2013.
\newblock URL: \url{http://arxiv.org/abs/1302.5401}, \href
  {http://arxiv.org/abs/1302.5401} {\path{arXiv:1302.5401}}.

\bibitem{DBLP:conf/soda/ParterP14}
Merav Parter and David Peleg.
\newblock Fault tolerant approximate {BFS} structures.
\newblock In Chandra Chekuri, editor, {\em Proceedings of the Twenty-Fifth
  Annual {ACM-SIAM} Symposium on Discrete Algorithms, {SODA} 2014, Portland,
  Oregon, USA, January 5-7, 2014}, pages 1073--1092. {SIAM}, 2014.
\newblock \href {https://doi.org/10.1137/1.9781611973402.80}
  {\path{doi:10.1137/1.9781611973402.80}}.

\bibitem{DBLP:conf/spaa/ParterP15}
Merav Parter and David Peleg.
\newblock Fault tolerant {BFS} structures: {A} reinforcement-backup tradeoff.
\newblock In Guy~E. Blelloch and Kunal Agrawal, editors, {\em Proceedings of
  the 27th {ACM} on Symposium on Parallelism in Algorithms and Architectures,
  {SPAA} 2015, Portland, OR, USA, June 13-15, 2015}, pages 264--273. {ACM},
  2015.
\newblock \href {https://doi.org/10.1145/2755573.2755590}
  {\path{doi:10.1145/2755573.2755590}}.

\bibitem{RaghavendraT12}
Prasad Raghavendra and Ning Tan.
\newblock Approximating csps with global cardinality constraints using {SDP}
  hierarchies.
\newblock In {\em Symposium on Discrete Algorithms {(SODA)}}, pages 373--387.
  {SIAM}, 2012.

\bibitem{DBLP:journals/talg/RodittyZ12}
Liam Roditty and Uri Zwick.
\newblock Replacement paths and \emph{k} simple shortest paths in unweighted
  directed graphs.
\newblock {\em {ACM} Trans. Algorithms}, 8(4):33:1--33:11, 2012.
\newblock \href {https://doi.org/10.1145/2344422.2344423}
  {\path{doi:10.1145/2344422.2344423}}.

\bibitem{DBLP:conf/stoc/Solomon14}
Shay Solomon.
\newblock From hierarchical partitions to hierarchical covers: optimal
  fault-tolerant spanners for doubling metrics.
\newblock In David~B. Shmoys, editor, {\em Symposium on Theory of Computing,
  {STOC} 2014, New York, NY, USA, May 31 - June 03, 2014}, pages 363--372.
  {ACM}, 2014.
\newblock \href {https://doi.org/10.1145/2591796.2591864}
  {\path{doi:10.1145/2591796.2591864}}.

\bibitem{Swamy04}
Chaitanya Swamy.
\newblock Correlation clustering: Maximizing agreements via semidefinite
  programming.
\newblock In {\em Symposium on Discrete Algorithms}, SODA, pages 526--527,
  2004.

\bibitem{DBLP:journals/talg/SwamyS08}
Chaitanya Swamy and David~B. Shmoys.
\newblock Fault-tolerant facility location.
\newblock {\em {ACM} Trans. Algorithms}, 4(4):51:1--51:27, 2008.
\newblock \href {https://doi.org/10.1145/1383369.1383382}
  {\path{doi:10.1145/1383369.1383382}}.

\bibitem{DBLP:journals/siamcomp/TrevisanSSW00}
Luca Trevisan, Gregory~B. Sorkin, Madhu Sudan, and David~P. Williamson.
\newblock Gadgets, approximation, and linear programming.
\newblock {\em {SIAM} J. Comput.}, 29(6):2074--2097, 2000.
\newblock \href {https://doi.org/10.1137/S0097539797328847}
  {\path{doi:10.1137/S0097539797328847}}.

\bibitem{DBLP:journals/informs/ZhouZTHD18}
Jiao Zhou, Zhao Zhang, Shaojie Tang, Xiaohui Huang, and Ding{-}Zhu Du.
\newblock Breaking the \emph{O}(ln \emph{n}) barrier: An enhanced approximation
  algorithm for fault-tolerant minimum weight connected dominating set.
\newblock {\em {INFORMS} J. Comput.}, 30(2):225--235, 2018.
\newblock \href {https://doi.org/10.1287/ijoc.2017.0775}
  {\path{doi:10.1287/ijoc.2017.0775}}.

\end{thebibliography}
\end{document}